\documentclass[aip, amsmath, amssymb, reprint, onecolumn]{revtex4-1}

\usepackage{amsfonts}
\usepackage{mathtools}
\usepackage{amsthm}
\usepackage{amssymb}
\usepackage{dsfont}
\usepackage{hyperref}
\usepackage{setspace}

\usepackage{graphicx}
\usepackage{dcolumn}
\usepackage{bm}
\usepackage{xfrac}

\usepackage[utf8]{inputenc}
\usepackage[T1]{fontenc}
\usepackage{mathptmx}
\usepackage{etoolbox}

\newtheorem{lemma}{Lemma}
\newtheorem{theorem}{Theorem}

\newtheorem{proposition}{Proposition}

\newcommand{\HF}{\mathrm{HF}}
\newcommand{\rHF}{\mathrm{rHF}}
\newcommand{\TF}{\mathrm{TF}}
\newcommand{\DFT}{\mathrm{LL}}

\newcommand{\R}{\mathbb{R}}
\newcommand{\N}{\mathbb{N}}
\newcommand{\C}{\mathbb{C}}

\newcommand{\dd}{\mathrm{d}}

\let\Re\relax
\let\Im\relax

\DeclareMathOperator{\DO}{DO}
\DeclareMathOperator{\den}{den}
\DeclareMathOperator{\Tr}{Tr}
\DeclareMathOperator{\argmax}{argmax}

\DeclareMathOperator{\Ran}{Ran}
\DeclareMathOperator{\Re}{Re}
\DeclareMathOperator{\Im}{Im}

\DeclarePairedDelimiter{\abs}{\lvert}{\rvert}
\DeclarePairedDelimiter{\norm}{\lVert}{\rVert}
\DeclarePairedDelimiter{\inner}{\langle}{\rangle}

\def\[#1\]{%
  \begin{equation}#1\end{equation}
}%

\makeatletter
\def\@email#1#2{%
  \endgroup
  \patchcmd{\titleblock@produce}
    {\frontmatter@RRAPformat}
    {%
      \frontmatter@RRAPformat{%
        \produce@RRAP{*#1\href{mailto:#2}{#2}}%
      }%
      \frontmatter@RRAPformat
    }
    {}%
    {}%
}
\makeatother

\begin{document}
  \title[A semiclassical limit of reduced Hartree--Fock theory at positive temperature]{A semiclassical limit of reduced Hartree--Fock theory at positive temperature}
  \author{D. Shillingford}
  \email{d.shillingford@mail.utoronto.ca}
  \affiliation{Department of Mathematics, University of Toronto}
  \date{\today}

  \begin{abstract}
    It is established that the reduced Hartree--Fock grand potential at positive temperature converges in the semiclassical limit to a Thomas--Fermi grand potential.
    In particular, a density--potential dual representation of the limiting functional is proposed for a general class of fermionic entropy functions.
    The corresponding dual problem has no duality gap and admits a maximizer.
    The resulting grand potential converges to the established zero-temperature grand potential as the temperature tends to zero and, for the Fermi--Dirac entropy, agrees with the standard positive-temperature Thomas--Fermi free-energy functional.
    The proof also establishes a semiclassical trace formula that is uniform over locally precompact families of potentials subject to a common confinement condition.
  \end{abstract}

  \maketitle

  \section{Introduction}

  The ground-state energy of a system of $n$ fermions subject to an external potential $v_n$, and interacting through a pair potential $w_n$,
  both perhaps depending on $n$, is
  \[ \label{manybodygse} \mathfrak{e}_n = \inf_{\substack{\Psi \in \mathfrak{H} \\ \norm{\Psi} = 1}} \inner[\Big]{\Psi, \mathfrak{h}_n \Psi}_\mathfrak{H} \]
  where $\mathfrak{H}$ is the subspace of $L^2(\R^{3n}; \C)$ that is antisymmetric with respect to particle exchange,
  and
  \[ \mathfrak{h}_n = \sum_{k = 1}^n \left( -\Delta_k + v_n(\hat x_k) \right) + \frac{1}{2} \sum_{i \neq j}^n w_n(\hat x_i - \hat x_j) \]
  is the Hamiltonian of the system.
  It is of significant interest to determine when $\mathfrak{e}_n$ can be approximated by simpler effective models.

  One approximation is the Hartree--Fock theory\cite{Hartree, Fock} obtained by restricting the variational problem in \eqref{manybodygse} to Slater determinants.
  Written in terms of one-body density matrices, Slater determinants correspond to orthogonal projections in $L^2(\R^3; \C)$ of rank $n$.
  After a standard relaxation from rank $n$ projections to fermionic density operators of trace $n$, the Hartree--Fock energy is given by
  \[ \mathfrak{e}^\HF_n = \inf_{\substack{\gamma \in \DO(\mathcal{H}) \\ \Tr \gamma = n}} \left[ \Tr \left( -\Delta + v_n(\hat x) \right) \gamma + \frac{1}{2} \iint w_n(x - y) \left[ \rho_\gamma(x) \rho_\gamma(y) - \abs{\gamma(x, y)}^2 \right] \, \dd x \, \dd y \right] \]
  where $\DO{(\mathcal{H})} = \{ 0 \leq \gamma \leq \mathds{1} : \gamma \in \mathfrak{S}^1{(\mathcal{H})} \}$ and $\mathcal{H} = L^2(\R^3; \C)$.

  A foundational rigorous analysis of the Hartree--Fock model for atoms and molecules was given by Lieb and Simon\cite{LiebSimonTF, LiebSimonHF}, and later expanded by Lions\cite{LionsHF}.
  In particular, Lieb and Simon proved that the Hartree--Fock energy is asymptotically equal to the ground-state energy in the heavy-atom regime where the number of electrons $n \to \infty$ while the corresponding atom or molecule stays neutral.
  After rescaling, this limit is an example of a mean-field scaling in which $v_n$ and $w_n$ are of orders $n^{2/3}$ and $n^{-1/3}$ respectively;
  see Ref. \onlinecite{BenedikterPortaSchlein} for a modern presentation.

  Closely related is the reduced Hartree--Fock theory, obtained by omitting the exchange term, i.e.
  \[ \mathfrak{e}^\rHF_n = \inf_{\substack{\gamma \in \DO(\mathcal{H}) \\ \Tr \gamma = n}} \left[ \Tr \left( -\Delta + v_n(\hat x) \right) \gamma + \frac{1}{2} \iint \rho_\gamma(x) w_n(x - y) \rho_\gamma(y) \, \dd x \, \dd y \right] \, . \]
  The resulting reduced Hartree--Fock energy functional is convex, and is therefore easier to analyze than the full Hartree--Fock energy functional, while agreeing with the (full) Hartree--Fock energy to leading order in $n$.
  As a result, it has been studied as a surrogate for the full model\cite{Solovej}.

  Another effective theory is the Thomas--Fermi model\cite{Thomas, Fermi}, whose variational theory and heavy-atom asymptotics were established rigorously by Lieb and Simon\cite{LiebSimonTF}.
  The Thomas--Fermi energy is
  \[ \mathfrak{e}^\TF_n = \inf_{\substack{\norm{\rho}_1 = n \\ \rho \geq 0}} \left[ \frac{3}{5} (6 \pi^2)^{2 / 3} \int \rho(x)^{5 / 3} \, \dd x +  \int v_n(x) \rho(x) \, \dd x + \frac{1}{2} \iint \rho(x) w_n(x - y) \rho(y) \, \dd x \, \dd y \right] \, , \]
  and it can be deduced from the work of Lieb and Simon\cite{LiebSimonTF} in the heavy-atom regime, and more generally from the work of Fournais, Lewin and Solovej\cite{Fournais} in the mean-field regime that
  \[ \label{energyconv} \mathfrak{e}^\HF_n \equiv \mathfrak{e}^\rHF_n \equiv \mathfrak{e}^\TF_n \bmod o(n^{5/3}) \, . \]

  In this paper, we examine the relation between reduced Hartree--Fock theory and the Thomas--Fermi model in the setting of the grand canonical ensemble at positive temperature.
  In the grand canonical ensemble, the ground-state energy is replaced by the grand potential and in analogy with \eqref{energyconv},
  we show that the reduced Hartree--Fock grand potential is asymptotically equal to an appropriate Thomas--Fermi grand potential.

  We now introduce the notation needed to state our results.
  Let $\mathcal{H} = L^2(\R^3; \C)$ and define
  \[ \operatorname{DO}{(\mathcal{H})} = \left\{ \gamma = \gamma^\ast \in \mathfrak{S}_1(\mathcal{H}) : 0 \leq \gamma \leq \mathds{1} \right\} \, . \]

  For $\rho \in L^1(\R^3)$ and $\mu \in L^\infty(\R^3)$ we write
  \[ \inner{\mu, \rho} = \int \mu(x) \rho(x) \, \dd x \, . \]

  For $\rho, \eta \in L^1(\R^3)$ and $w \in L^\infty(\R^3)$ we write
  \[ \inner{\rho, \eta}_w = \iint \rho(x) w(x - y) \eta(y) \, \dd x \, \dd y \, . \]

  For $\gamma \in \DO{(\mathcal{H})}$ we denote by $\rho_\gamma$ or $\den \gamma$ the unique element of $L^1(\R^3)$ such that
  \[ \operatorname{Tr}{(\mu(\hat x) \, \gamma)} = \int \mu(x) \rho_\gamma(x) \, \mathrm{d} x = \inner{ \mu, \rho_\gamma } \]
  for every $\mu \in L^\infty(\mathbb{R}^3)$.

  We call $v : \R^3 \to \R$ \emph{confining} if $v(x) \to \infty$ as $\abs{x} \to \infty$.

  We call $s : [0, 1] \to \R$ an \emph{entropy function} if it is continuous and concave, and satisfies $s(0) = s(1) = 0$.
  Its concave conjugate $s_\ast : \R \to \R$ is defined by
  \[ s_\ast(h) = \inf_{0 \leq r \leq 1} \left( rh - s(r) \right) \, . \]

  We call $w \in L^\infty(\R^3)$ \emph{repulsive} if $w \geq 0$, and of \emph{positive type} if $\inner{-, -}_w$ is positive definite.
  In particular, if $w \in \mathcal{S}(\R^3)$ then $w$ is of positive type if and only if $\hat w > 0$ almost everywhere.

  For a function $f$, we denote by $f_+$ and $f_-$ respectively the positive and negative parts of $f$ with the convention that $f = f_+ - f_-$, i.e. $f_\pm \geq 0$.

  Lastly, when $H$ is a semibounded operator on $\mathcal{H}$ with compact resolvent, the expression
  \[ \Tr{(H \gamma)} = \sum_{n = 1}^\infty \lambda_n \inner{u_n, \gamma u_n} \, , \]
  where $(u_n)$ is an orthonormal basis of eigenvectors of $H$ satisfying $H u_n = \lambda_n u_n$, takes a well-defined extended value in $(-\infty, \infty]$
  for all $\gamma \in \DO{(\mathcal{H})}$.
  In this case, we define
  \[ \DO_H{(\mathcal{H})} = \{ \gamma \in \DO{(\mathcal{H})} : \Tr\left( H \gamma \right) < \infty \} \, . \]

  The Thomas--Fermi grand potential is defined by
  \[ \label{ftfdef} \Omega^\TF(\zeta, \beta) = \inf_{\rho \in L^1(\R^3)} \Omega^\TF(\rho; \zeta, \beta) \, , \]
  where
  \[ \label{ftfrhodef} \Omega^\TF(\rho; \zeta, \beta) = \frac{1}{2} \inner{ \rho, \rho }_w + \sup_{\mu \in L^\infty(\R^3)} \left[  \frac{\beta^{-1}}{(2 \pi)^3} \iint s_\ast(\beta(\abs{p}^2 + v(x) + \mu(x) - \zeta)) \, \mathrm{d} x \, \mathrm{d} p - \inner{ \mu, \rho } \right] \, , \]
  $\beta$ is the inverse temperature of the system and $\zeta$ is the chemical potential of the system.

  Our main result has the following equivalent mean-field and semiclassical formulations.
  \begin{theorem}
    \label{mainthm}
    Fix $\beta > 0$ and $\zeta \in \R$. Let $v \in L^{5/2}_\mathrm{loc.}(\R^3)$ be confining, $w \in \mathcal{S}(\R^3)$ be even, repulsive and of positive type, and $s : [0, 1] \to \R$ be an entropy function
    such that $s_\ast \in C^2(\R)$ and the compatibility condition
    \[ \iint s_\ast(\beta(\abs{p}^2 + v(x) - \zeta)) \, \dd x \, \dd p > -\infty \]
    holds.

    Then $\Omega^\TF(\zeta, \beta)$ is finite, and the following equivalent formulations hold:
    \begin{enumerate}
      \item[(a)] \emph{Mean-field formulation.}
        Let $\beta_n = n^{-2/3} \beta$, $\zeta_n = n^{2/3} \zeta$, $v_n = n^{2/3} v$, $w_n = n^{-1/3} w$, $H_n = -\Delta + v_n(\hat x)$, and $\DO_n(\mathcal{H}) = \DO_{H_n}(\mathcal{H})$.
        The reduced Hartree--Fock grand potential
        \[ \Omega^\rHF(\zeta_n, \beta_n) = \inf_{\gamma \in \DO_n{(\mathcal{H})}} \left[ \Tr \left( H_n - \zeta_n \right) \gamma - \beta_n^{-1} \Tr s(\gamma) + \frac{1}{2} \inner{\rho_\gamma, \rho_\gamma}_{w_n} \right] \]
        is finite for each $n$, and
        \[ n^{-5/3} \Omega^\rHF(\zeta_n, \beta_n) \to \Omega^\TF(\zeta, \beta) \, . \]

      \item[(b)] \emph{Semiclassical formulation.}
        Let $H_\kappa = -\kappa^2 \Delta + v(\hat x)$ and $\DO_\kappa(\mathcal{H}) = \DO_{H_\kappa}(\mathcal{H})$.
        The rescaled reduced Hartree--Fock grand potential
        \[ \label{fksdef} \Omega^\rHF_\kappa(\zeta, \beta) = \inf_{\gamma \in \DO_\kappa{(\mathcal{H})}} \left[ \Tr \left( H_\kappa - \zeta \right) \gamma - \beta^{-1} \Tr s(\gamma) + \frac{\kappa^3}{2} \inner{\rho_\gamma, \rho_\gamma}_{w} \right] \equiv \inf_{\gamma \in \DO_\kappa{(\mathcal{H})}} \Omega^\rHF_\kappa(\gamma; \zeta, \beta) \]
        is finite for each $\kappa > 0$, and
        \[ \label{thmclose} \kappa^3 \Omega^\rHF_\kappa(\zeta, \beta) \to \Omega^\TF(\zeta, \beta) \, . \]
        In addition, if $\mu^\rHF_\kappa$ and $\mu^\TF$ are the corresponding mean-field potentials from Proposition \ref{ksdftequiv} and Proposition \ref{swaptf}, respectively, then
        \[ \mu^\rHF_\kappa \xrightarrow{\kappa \to 0} \mu^\TF \qquad \text{in } H^\sigma(\R^3) \text{ for all } \sigma \geq 0 \, . \]
    \end{enumerate}

    The two formulations are related by $\kappa = n^{-1/3}$, for which
    \[ \Omega^\rHF(\zeta_n, \beta_n) = n^{2/3} \Omega^\rHF_{n^{-1/3}}(\zeta, \beta) \, . \]
  \end{theorem}

  \noindent \emph{Remarks.}
  \begin{enumerate}
    \item The mean-field scaling in part (a) balances all contributions to the grand potential.
    For states with particle number of order $n$, the kinetic, external-potential, interaction, entropy and chemical-potential terms are all of order $n^{5/3}$:
    the energy, temperature and chemical-potential scales are of order $n^{2/3}$ per particle, while $w_n = n^{-1/3} w$ compensates for the order $n^2$
    interacting pairs.
    Thus both thermal and interaction effects persist in the limiting Thomas--Fermi theory.
    \item By the scaling relation in the theorem, the convergence of the mean-field potentials in part (b) can equivalently be restated in the mean-field scaling of part (a).
  \end{enumerate}

  Positive-temperature Thomas--Fermi theory with the standard Fermi--Dirac entropy
  \[ s(r) = -r \log{r} - (1 - r) \log{(1 - r)} \]
  has long been studied in both the physics and mathematics literature.
  Early physics treatments include the foundational work of Marshak and Bethe\cite{MarshakBethe}
  and that of Feynman, Metropolis and Teller\cite{FeynmanMetropolisTeller}.
  In the mathematics literature, Narnhofer and Sewell\cite{NarnhoferSewell} considered self-gravitating fermions in the thermodynamic limit,
  while Narnhofer and Thirring\cite{NarnhoferThirring} established the asymptotic exactness of positive-temperature Thomas--Fermi theory in a grand-canonical heavy-atom regime.
  More recently, Lewin, Madsen and Triay\cite{Lewin} studied the coupled mean-field and semiclassical limit in the closely related canonical ensemble.

  Beyond the Fermi--Dirac case, positive-temperature Thomas--Fermi models based on nonextensive entropies, including the Tsallis entropy\cite{Tsallis, MartinenkoShivamoggi}, have also been considered in the physics literature.
  More generally, Chavanis\cite{Chavanis} considered an arbitrary phase-space entropy in a canonical self-gravitating model and derived a free-energy functional of the spatial density.
  After matching conventions, the local ideal-gas term in his functional agrees with the one induced by \eqref{ftfrhodef}, although his interaction and ensemble are different.

  Nguyen\cite{Nguyen} studied the grand-canonical mean-field problem at zero temperature, corresponding formally to $\beta=\infty$, or equivalently to $s=0$.

  To the author's knowledge, the grand-canonical positive-temperature Thomas--Fermi problem has not previously been formulated and analyzed for a general class of fermionic entropy functions using the density--potential dual representation \eqref{ftfrhodef}.
  This formulation is particularly suited to the semiclassical comparison with the reduced Hartree--Fock grand potential carried out here.

  Additionally, the following secondary results support the consistency of \eqref{ftfdef} with previous results.
  As noted above, the relationship between the Thomas--Fermi grand potential and both the reduced Hartree--Fock and Hartree--Fock grand potentials was considered at zero temperature by Nguyen\cite{Nguyen},
  who defined
  \[ \omega^\rHF_\kappa(\zeta) = \inf_{\substack{\gamma \in \DO(\mathcal{H})}} \omega^\rHF_{\kappa} (\gamma; \zeta) \, , \]
  where
  \[ \omega^\rHF_{\kappa} (\gamma; \zeta) = \Tr \left( -\kappa^2 \Delta + v(\hat x) - \zeta \right) \gamma + \frac{\kappa^3}{2} \iint \rho_\gamma(x) w(x - y) \rho_\gamma(y)  \, \dd x \, \dd y \, ,  \]
  and
  \[ \label{nguyenftfdef} \omega^\TF(\zeta) = \inf_{\rho \geq 0} \left[ \int \left[ \frac{3}{5} (6 \pi^2)^{2 / 3} \rho(x)^{5 / 3} + \left[ v(x) - \zeta \right] \rho(x) \right] \, \dd x + \frac{1}{2} \iint \rho(x) w(x - y) \rho(y) \, \dd x \, \dd y \right] \, . \]
  Indeed, it was shown by Nguyen that for each fixed $\zeta \in \R$
  \[ \label{gpconv} \kappa^3 \omega^\rHF_\kappa(\zeta) \xrightarrow{\kappa \to 0} \omega^\TF(\zeta) \, . \]
  Our grand potential is consistent with that of Nguyen in the following sense.

  \begin{theorem}
    \label{mainthmtwo}

    Fix $\beta_0 > 0$ and $\zeta \in \R$. Let $v \in L^{5/2}_\mathrm{loc.}(\R^3)$ be confining, $w \in \mathcal{S}(\R^3)$ be even, repulsive and of positive type, and $s : [0, 1] \to \R$ be an entropy function
    such that the compatibility condition
    \[ \iint s_\ast(\beta_0(\abs{p}^2 + v(x) - \zeta)) \, \dd x \, \dd p > -\infty \]
    holds.

    Then $\omega^\TF(\zeta)$ is finite, $\Omega^\TF(\zeta, \beta)$ is finite for all $\beta \geq \beta_0$,  and
    \[ \Omega^\TF(\zeta, \beta) \xrightarrow{\beta \to \infty} \omega^\TF(\zeta) \, . \]
  \end{theorem}

  Our grand potential is also consistent with the Thomas--Fermi free energy considered by Lewin, Madsen and Triay\cite{Lewin}
  for the Fermi--Dirac entropy.
  Specifically, Lewin, Madsen and Triay defined
  \[ \label{freeenergytfdef} F^\TF(\beta) = \inf_{\substack{\norm{\rho}_1  = 1 \\ \rho \geq 0}} \int f^\ast(\rho(x)) + v(x) \rho(x) \, \dd x + \frac{1}{2} \inner{\rho, \rho}_w \]
  where $f^*(0) = 0$ and
  \[ \label{fstardef} f^*(r) = r h(r) - \frac{\beta^{-1}}{(2 \pi)^3} \int \log \left( 1 + \exp{\left( -\beta(\abs{p}^2 - h(r) \right)} \right) \, \dd p \, , \]
  with $h(r)$ the unique solution to the implicit equation
  \[ r = \frac{1}{(2\pi)^3} \int \left( 1 + \exp{(\beta (\abs{p}^2 - h(r)))} \right)^{-1} \, \dd p \]
  for $r > 0$.
  Indeed, we have the following.
  \begin{theorem}
    \label{mainthmthree}

    Fix $\beta > 0$ and $\zeta \in \R$. Let $v \in L^{5/2}_\mathrm{loc.}(\R^3)$ be confining, $w \in \mathcal{S}(\R^3)$ be even, repulsive and of positive type, and $s : [0, 1] \to \R$ such that
    \[ s(r) = -r \log{r} - (1 - r) \log {(1 - r)} \]
    and
    \[ \int \exp(-\beta v(x)_+) \, \dd x < \infty \, . \]

    Then
    \[ \Omega^\TF(\zeta, \beta) = \inf_{\substack{\rho \in L^1(\R^3) \\ \rho \geq 0}} \int f^\ast(\rho(x)) + \left[ v(x) - \zeta \right] \rho(x) \, \dd x + \frac{1}{2} \inner{\rho, \rho}_w \]
    where $f^*$ is defined as in \eqref{fstardef}.
  \end{theorem}

  \section{Proof of Theorem \ref{mainthm}}

  In this section we give a proof of Theorem \ref{mainthm} while deferring some of the more technical steps to lemmas which will be proven further in the paper.

  \begin{proof}[Proof of Theorem \ref{mainthm}]
    We prove the semiclassical formulation in part (b); the mean-field formulation in part (a) then follows from the scaling relation in the theorem.
    In this formulation, since $v - \zeta$ is an admissible external potential whenever $v$ is, and
    similarly $\beta^{-1} s$ is an admissible entropy function whenever $s$ is, we can take $\zeta = 0$ and $\beta = 1$ without loss of generality.
    As such, we drop $\zeta$ and $\beta$ from the notation throughout this proof for simplicity, e.g., $\Omega^\TF(\rho; \zeta, \beta)$ becomes $\Omega^\TF(\rho)$
    and $\Omega^\TF(\zeta, \beta)$ becomes $\Omega^\TF$.

    We first show the finiteness of $\Omega^\TF$.
    Let
    \[ \label{ftfrhomudef} \Omega^\TF(\rho, \mu) = \frac{1}{(2\pi)^3} \iint s_\ast(\abs{p}^2 + v(x) + \mu(x)) \, \dd p \, \dd x + \frac{1}{2} \inner{\rho, \rho}_w - \inner{\mu, \rho} \, , \]
    so that $\Omega^\TF(\rho) = \sup_{\mu \in L^\infty(\R^3)} \Omega^\TF(\rho, \mu)$, and
    let $\Omega^\TF(\mu) = \inf_{\rho \in L^1(\R^3)} \Omega^\TF(\rho, \mu)$ be the corresponding dual functional.

    It then follows from Proposition \ref{swaptf} that we have no duality gap, i.e.,
    \[ \Omega^\TF = \sup_{\mu \in L^\infty(\R^3)} \Omega^\TF(\mu) \]
    and thus $\Omega^\TF(\mu = 0) \leq \Omega^\TF \leq \Omega^\TF(\rho = 0)$.

    \[ \Omega^\TF(\rho = 0) = \sup_{\mu \in L^\infty(\R^3)} \frac{1}{(2\pi)^3} \iint s_\ast(\abs{p}^2 + v(x) + \mu(x)) \, \dd p \, \dd x \leq 0 \]
    since $s_\ast(h) = \inf_{0 \leq r \leq 1} ( rh - s(r) ) \leq -s(0) = 0$.

    By computation, and since $\hat w > 0$,
    \begin{align}
      \Omega^\TF(\mu = 0) &= \inf_{\rho \in L^1(\R^3)} \left[ \frac{1}{(2 \pi)^3} \iint s_\ast(\abs{p}^2 + v(x)) \, \mathrm{d} x \, \mathrm{d} p + \frac{1}{2} \inner{ \rho, \rho }_w \right] \\
      &= \frac{1}{(2 \pi)^3} \iint s_\ast(\abs{p}^2 + v(x)) \, \mathrm{d} x \, \mathrm{d} p > -\infty
    \end{align}
    by hypothesis.
    We therefore conclude that $\Omega^\TF$ is finite.

    Next, we show the finiteness of $\Omega_\kappa^\rHF$ for $\kappa > 0$.
    By Fenchel's inequality (see Lemma \ref{operatorfenchel} for details), for each $\lambda \in \mathbb{R}$ we have
    \[ \lambda \gamma - s(\gamma) \geq s_\ast(\lambda) \mathds{1} \, . \]
    By Lemma \ref{semiclassicalbound}, $H_\kappa$ has compact resolvent and therefore there exists an orthonormal basis $(u_n)$ of $\mathcal{H}$ such that for each $n$,
    $u_n$ is an eigenvector of $H_\kappa$ with corresponding eigenvalue $\lambda_n$.
    Evaluating the trace on this basis, it follows that
    \begin{align}
      \Omega_\kappa^\rHF(\gamma) &= \sum_{n = 1}^\infty \inner{ u_n, \left[ \lambda_n \gamma - s(\gamma) - s_\ast(\lambda_n) \right] u_n } + \Tr s_\ast (H_\kappa) + \frac{\kappa^3}{2} \inner{ \rho_\gamma, \rho_\gamma }_w \\
      &\geq \Tr s_\ast(H_\kappa) + \frac{\kappa^3}{2} \inner{ \rho_\gamma, \rho_\gamma }_w \, .
    \end{align}
    Since $\hat w > 0$, we have $0 = \Omega^\rHF_\kappa(\gamma = 0) \geq \Omega_\kappa^\rHF \geq \Tr s_\ast(H_\kappa)$.
    We then conclude by Lemma \ref{semiclassicalbound} that $\Omega_\kappa^\rHF$ is finite as
    \[ \Tr s_\ast(H_\kappa) \gtrsim \frac{1}{(2\pi \kappa)^3} \iint s_\ast(\abs{p}^2 + v(x)) \, \dd x \, \dd p > -\infty \, . \]

    It remains to show the convergence $\eqref{thmclose}$.
    To that end, we introduce auxiliary functionals $\Omega_\kappa^\DFT(\rho, \mu)$ and $\Omega_\kappa^\DFT(\mu)$ where
    \[ \label{fdftrhomudef} \Omega_\kappa^\DFT(\rho, \mu) = \Tr {s_\ast(H_\kappa + \mu(\hat x))} + \frac{\kappa^3}{2} \inner{ \rho, \rho }_w - \inner{ \mu, \rho }\, , \]
    and
    \[ \label{fdftmudef} \Omega_\kappa^\DFT(\mu) = \inf_{\rho \in L^1(\R^3)} \Omega_\kappa^\DFT(\rho, \mu) \, . \]

    By Proposition \ref{ksdftequiv}, the functional $\mu\mapsto\Omega_\kappa^\DFT(\mu)$ achieves a global maximum at $\mu^\DFT_\kappa$
    and $\Omega^\rHF_\kappa = \Omega^\DFT_\kappa(\mu^\DFT_\kappa)$.
    In fact, along every minimizing sequence $(\gamma^\rHF_{\kappa, n})$ for $\Omega_\kappa^\rHF$ we have
    \[ \kappa^3 w \ast \operatorname{den} \gamma^\rHF_{\kappa, n} \xrightarrow{n \to \infty} \mu^\rHF_\kappa = \mu^\DFT_\kappa \, , \]
    and it therefore suffices to show that
    \[ \kappa^3 \Omega_\kappa^\DFT(\mu_\kappa^\rHF) \xrightarrow{\kappa \to 0} \Omega^\TF \, . \]

    An upper bound for the difference between $\kappa^3 \Omega_\kappa^\DFT(\mu_\kappa^\rHF)$ and $\Omega^\TF$ is computed as follows:
    \begin{align}
      \kappa^3 \Omega_\kappa^\DFT(\mu_\kappa^\rHF) - \Omega^\TF &\leq \kappa^3 \Omega_\kappa^\DFT(\mu_\kappa^\rHF) - \Omega^\TF(\mu_\kappa^\rHF) \\[1ex]
      &= \sup_{\rho \in L^1(\R^3)} \left[ \kappa^3 \Omega_\kappa^\DFT(\mu_\kappa^\rHF) - \Omega^\TF(\rho, \mu_\kappa^\rHF) \right] \\
      &= \sup_{\rho \in L^1(\R^3)} \left[ \kappa^3 \Omega_\kappa^\DFT(\mu_\kappa^\rHF) - \Omega^\TF(\kappa^3 \rho, \mu_\kappa^\rHF) \right] \\
      &\leq \sup_{\rho \in L^1(\R^3)} \left[ \kappa^3 \Omega_\kappa^\DFT(\rho, \mu_\kappa^\rHF) - \Omega^\TF(\kappa^3 \rho, \mu_\kappa^\rHF) \right] \\
      &= \kappa^3 \Tr s_\ast(H_\kappa + \mu_\kappa^\rHF(\hat x)) - \frac{1}{(2 \pi)^3} \iint s_\ast(\abs{p}^2 + v(x) + \mu_\kappa^\rHF(x)) \, \dd x \, \dd p \, .
    \end{align}

    Similarly, by Proposition \ref{swaptf}, the functional $\mu \mapsto \Omega^\TF(\mu)$ achieves a global maximum at $\mu^\TF$,
    and we use this maximizer to derive a complementary lower bound as follows:
    \begin{align}
      \kappa^3 \Omega_\kappa^\DFT(\mu_\kappa^\rHF) - \Omega^\TF &= \kappa^3 \Omega_\kappa^\DFT(\mu_\kappa^\rHF) - \Omega^\TF(\mu^\TF) \\[1ex]
      &\geq \kappa^3 \Omega_\kappa^\DFT(\mu^\TF) - \Omega^\TF(\mu^\TF) \\[1ex]
      &= \inf_{\rho \in L^1(\R^3)} \left[ \kappa^3 \Omega_\kappa^\DFT(\rho, \mu^\TF) - \Omega^\TF(\mu^\TF) \right] \\
      &\geq \inf_{\rho \in L^1(\R^3)} \left[ \kappa^3 \Omega_\kappa^\DFT(\rho, \mu^\TF) - \Omega^\TF(\kappa^3 \rho, \mu^\TF) \right] \\
      &= \kappa^3 \Tr s_\ast(H_\kappa + \mu^\TF(\hat x)) - \frac{1}{(2 \pi)^3} \iint s_\ast(\abs{p}^2 + v(x) + \mu^\TF(x)) \, \dd x \, \dd p \, .
    \end{align}

    \eqref{thmclose} then follows if
    \[ \label{diffmuks} \left[ \kappa^3 \Tr s_\ast(H_\kappa + \mu_\kappa^\rHF(\hat x)) - \frac{1}{(2 \pi)^3} \iint s_\ast(\abs{p}^2 + v(x) + \mu_\kappa^\rHF(x)) \, \dd x \, \dd p \right] \xrightarrow{\kappa \to 0} 0  \, , \]
    and
    \[ \label{diffmutf} \kappa^3 \Tr s_\ast(H_\kappa + \mu^\TF(\hat x)) \xrightarrow{\kappa \to 0} \frac{1}{(2 \pi)^3} \iint s_\ast(\abs{p}^2 + v(x) + \mu^\TF(x)) \, \dd x \, \dd p \, . \]

    Using the decomposition
    \[ \Tr s_\ast{(H_\kappa + \mu(\hat x))} = \int s_\ast''(t) \Tr{\left( H_\kappa + \mu(\hat x) - t \right)_-} \, \dd t \, , \]
    the convergence in \eqref{diffmutf} follows straightforwardly from the well-known semiclassical convergence of Riesz means\cite{Nam, Frank, Mikkelsen}.
    The convergence in \eqref{diffmuks} is more subtle because of the dependence of $\mu^\rHF_\kappa$ on $\kappa$.

    In both cases, we apply Proposition \ref{semiclassics}.
    For \eqref{diffmutf}, we apply the proposition to $\{ v_\kappa = v + \mu^\TF : \kappa > 0 \}$
    with $v_{\min} = v \leq v + \mu^\TF$ as $\mu^\TF \geq 0$ by Proposition \ref{swaptf}.
    Indeed, $\{ \left. v_\kappa \right|_K \} \subseteq L^{5/2}(K)$ is a singleton for each $K$ compact and thus compact.

    For \eqref{diffmuks}, we apply the proposition to $\{ v_\kappa = v + \mu_\kappa^\rHF : \kappa > 0 \}$
    with $v_{\min} = v \leq v + \mu^\rHF_\kappa$ as $\mu^\rHF_\kappa \geq 0$.
    It remains to verify that $\{ \left. v_\kappa \right|_K \}$ is relatively compact for each $K$ compact.
    To that end, it follows from Lemma \ref{fksmunice} that
    \[ \norm{\mu^\rHF_\kappa}_{H^s(\R^3)}^2 \lesssim_{w, s} -\kappa^3 \Omega_\kappa^\DFT(\mu^\rHF_\kappa) \leq -\kappa^3 \Omega^\DFT_\kappa(\mu = 0) \]
    where, as we have seen above, $\Omega^\DFT_\kappa(\mu = 0) = \Tr s_\ast(H_\kappa)$.
    We then conclude by Lemma \ref{semiclassicalbound} that
    \[ \norm{\mu^\rHF_\kappa}_{H^1(\R^3)}^2 \lesssim_{w} -\frac{1}{(2\pi)^3} \iint s_\ast(\abs{p}^2 + v(x)) \, \dd x \, \dd p \, . \]
    In particular, for each $K$ compact, we have
    \[ \norm{\mu^\rHF_\kappa}_{H^1(U)}^2 \lesssim_{w} -\frac{1}{(2\pi)^3} \iint s_\ast(\abs{p}^2 + v(x)) \, \dd x \, \dd p \]
    where $U$ is some bounded Lipschitz domain containing $K$, e.g. a sufficiently large ball.

    By the Rellich--Kondrashov Theorem, it then follows that $\{ \left. \mu^\rHF_\kappa \right|_U \}$ is relatively compact in $L^{5/2}(U)$
    and by restriction $\{ \left. \mu^\rHF_\kappa \right|_K \}$ is relatively compact in $L^{5/2}(K)$.
    To conclude, since addition is continuous in $L^{5/2}(K)$, $\{ \left. v_\kappa \right|_K \}$ is relatively compact in $L^{5/2}(K)$.
    The two convergence statements we seek then follow from Proposition \ref{semiclassics}.

    By Lemma \ref{infizernice} (cf. Lemma \ref{ftfmunice}), $\mu \mapsto \Omega^\TF(\mu)$ is strongly concave, i.e.,
    \[ \frac{1}{8} \int \abs{ \hat\mu_\kappa^\rHF(k) - \hat\mu^\TF(k) }^2 \, \hat w(k)^{-1} \, \dd k \leq \Omega^\TF\left( \frac{\mu_\kappa^\rHF + \mu^\TF}{2} \right) - \frac{\Omega^\TF(\mu_\kappa^\rHF) + \Omega^\TF(\mu^\TF)}{2} \leq \frac{1}{2} \left( \Omega^\TF - \Omega^\TF(\mu_\kappa^\rHF) \right) \, . \]
    We then have
    \[ 0 \leq \Omega^\TF - \Omega^\TF(\mu^\rHF_\kappa) = (\Omega^\TF - \kappa^3 \Omega_\kappa^\rHF) - (\kappa^3 \Omega_\kappa^\DFT(\mu^\rHF_\kappa) - \Omega^\TF(\mu^\rHF_\kappa)) \, . \]
    We have already established above that the first term goes to zero as $\kappa \to 0$, while the second term goes to zero by \eqref{diffmuks}.
    We therefore conclude that
    \[ \int \abs{ \hat\mu_\kappa^\rHF(k) - \hat\mu^\TF(k) }^2 \, \hat w(k)^{-1} \, \dd k \xrightarrow{\kappa \to 0} 0 \, . \]
    In particular, since
    \[ \norm{\mu^\rHF_\kappa - \mu^\TF}_{H^\sigma(\R^3)}^2 \leq \norm{\langle - \rangle^{2\sigma} \hat w}_{L^\infty(\R^3)} \int \abs{ \hat\mu_\kappa^\rHF(k) - \hat\mu^\TF(k) }^2 \, \hat w(k)^{-1} \, \dd k \, , \]
    we have $\mu^\rHF_\kappa \xrightarrow{\kappa \to 0} \mu^\TF$ in $H^\sigma(\R^3)$ for all $\sigma \geq 0$.
  \end{proof}

  \noindent \emph{Remark.} To motivate the choice of symbol $\Omega^\DFT_\kappa$
  one can show, at least formally, that
  \[ \Omega^\DFT_\kappa(\rho) = \sup_{\mu \in L^\infty(\R^3)} \Omega_\kappa^\DFT(\rho, \mu) \]
  is in fact the Levy--Lieb density functional for the reduced Hartree--Fock model, i.e.,
  \[ \Omega_\kappa^\DFT(\rho) = \inf_{\substack{\gamma \in \DO_\kappa(\mathcal{H}) \\ \rho_\gamma = \rho}} \Omega^\rHF_\kappa(\gamma) \]
  and thus $(\rho, \mu) \mapsto \Omega^\DFT_\kappa(\rho, \mu)$ can be seen as a saddle function for the minimization of the
  Levy--Lieb density functional, and $\mu \mapsto \Omega^\DFT_\kappa(\mu)$ is the corresponding dual objective function.

  \section{Overview of Finite Temperature Phenomena}

  Before proceeding to the proof of our remaining results, we briefly comment on some of the key new phenomena in the positive temperature case.

  Formally, since $0 \leq \gamma \leq \mathds{1}$ we have
  \[ \label{bathtub} \Tr \left( -\kappa^2 \Delta + v(\hat x) - \zeta \right) \gamma \geq -\Tr{[\zeta +\kappa^2 \Delta - v(\hat x)]_+} \]
  which can be used to establish several coercivity estimates for $\omega_\kappa^\HF(-; \zeta)$.

  In particular,
  \[ \label{zcoulomb} \omega_\kappa^\rHF(\gamma; \zeta) \geq \frac{\kappa^3}{2} \inner{\rho_\gamma, \rho_\gamma}_w - \Tr{ [\zeta + \kappa^2 \Delta - v(\hat x)]_+ }\, , \]
  \[ \label{zenergy} \omega_\kappa^\rHF(\gamma; \zeta) \geq \Tr { \left( -\kappa^2 \Delta + v(\hat x) - \zeta \right) \gamma } \, , \]
  while for all $\epsilon > 0$,
  \[ \label{zmass} \omega_\kappa^\rHF(\gamma; \zeta) = \frac{\kappa^3}{2} \inner{\rho_\gamma, \rho_\gamma} + \epsilon \Tr \gamma + \Tr \left( -\kappa^2 \Delta + v(\hat x) - \zeta - \epsilon \right) \gamma \geq \epsilon \Tr \gamma - \Tr{ [\zeta + \epsilon + \kappa^2 \Delta - v(\hat x)]_+ } \, . \]

  These estimates yield uniform control over the energy, interaction energy and particle number respectively along minimizing sequences.
  In particular, we draw attention to the fact that since the eigenvalues of $-\kappa^2 \Delta + v(\hat x)$ accumulate only at infinity
  \[ \Tr{[\zeta + \epsilon + \kappa^2 \Delta - v(\hat x)]_+} < \infty \]
  for all $\epsilon > 0$.

  In the finite temperature case considered in this paper, the analogue to the estimate \eqref{bathtub} is the Fenchel inequality
  \[ \Tr \left( -\kappa^2 \Delta + v(\hat x) - \zeta \right) \gamma - \beta^{-1} \Tr {s(\gamma)} \geq \beta^{-1} \Tr {s_\ast(\beta (-\kappa^2 \Delta + v(\hat x) - \zeta))} \, . \]
  It then follows that
  \[ \Omega_\kappa^\rHF(\gamma; \zeta, \beta) \geq \frac{\kappa^3}{2} \inner{\rho_\gamma, \rho_\gamma}_w + \beta^{-1} \Tr{ s_\ast(\beta(-\kappa^2 \Delta + v(\hat x) - \zeta))} \geq \beta^{-1} \Tr{ s_\ast(\beta(-\kappa^2 \Delta + v(\hat x) - \zeta))}  \, , \]
  and therefore we have a well-defined and finite grand potential exactly when
  \[ \label{ffinite} \Tr{ s_\ast(\beta(-\kappa^2 \Delta + v(\hat x) - \zeta))} > -\infty \, . \]
  The condition \eqref{ffinite} will hold only if the growth of $v$ is compatible
  with the decay of $s_\ast$, demonstrating a genuinely new phenomenon in the $\beta < \infty$ case.

  A more subtle new phenomenon is that the condition \eqref{ffinite} does not appear to be enough on its own to get coercivity of the energy or particle number.
  As an example, one may carry out a similar calculation to \eqref{zmass} to find that if
  \[ \Tr{ s_\ast(\beta(-\kappa^2 \Delta + v(\hat x) - \zeta - \epsilon))} > -\infty \]
  for some $\epsilon > 0$ we have
  \[ \Omega_\kappa^\rHF(\gamma; \zeta, \beta) \geq \epsilon \Tr \gamma + \beta^{-1} \Tr{ s_\ast(\beta(-\kappa^2 \Delta + v(\hat x) - \zeta - \epsilon))} \]
  but there exist systems such that for $\epsilon \geq 0$,
  \[ \Tr{ s_\ast(\beta(-\kappa^2 \Delta + v(\hat x) - \zeta - \epsilon))} > -\infty \iff \epsilon = 0 \, . \]
  We call such systems \emph{mass critical} at $(\zeta, \beta)$ and it is unclear how one might recover a coercivity estimate for the particle number in these instances.

  We remark that this phenomenon is also unique to the broadened class of entropy functions considered in this paper,
  as for the case $s(r) = -r \log r - (1 - r) \log (1 - r)$ we have
  \[ s_\ast(h) = - \log \left( 1 + \exp(-h) \right) \, . \]
  One can show that $h \mapsto s_\ast(h) \exp(h)$ is increasing with $s_\ast(h) \exp(h) \to 1$ as $h \to \infty$ and thus
  \[ \label{vnfreeenergybound} -\exp(-h) \leq s_\ast(h) \lesssim_{h_0} -\exp(-h) \]
  for all $h \geq h_0$ and therefore with $h_0 = \inf \sigma(-\kappa^2 \Delta + v(\hat x) - \zeta \mathds{1})$ we have
  \begin{align}
    \Tr{ s_\ast(\beta(-\kappa^2 \Delta + v(\hat x) - \zeta - \epsilon))} &\asymp_{h_0} -\Tr\, { \exp(-\beta(-\kappa^2 \Delta + v(\hat x) - \zeta - \epsilon))} \\
    &= -\exp(\beta \epsilon) \Tr\, { \exp(-\beta(-\kappa^2 \Delta + v(\hat x) - \zeta))} \\
    &\asymp_{h_0} \exp(\beta \epsilon) \Tr{ s_\ast(\beta(-\kappa^2 \Delta + v(\hat x) - \zeta))} \, .
  \end{align}
  For this particular choice of $s$, it then follows that $\Tr{ s_\ast(\beta(-\kappa^2 \Delta + v(\hat x) - \zeta - \epsilon))}$ is finite
  whenever $\Tr{ s_\ast(\beta(-\kappa^2 \Delta + v(\hat x) - \zeta))}$ is but this is not the case in general.

  Similarly, if
  \[ \Tr{ s_\ast(\tilde\beta(-\kappa^2 \Delta + v(\hat x) - \zeta))} > -\infty \]
  for some $\tilde \beta < \beta$ one can show
  \[ \Omega^\rHF_{\kappa} (\gamma; \zeta, \beta) \geq \left(1 - \frac{\tilde\beta}{\beta}\right) \Tr \left( -\kappa^2 \Delta + v(\hat x) - \zeta \right) \gamma + \beta^{-1} \Tr { s_\ast(\tilde\beta ( -\kappa^2 \Delta + v(\hat x) - \zeta ) } \]
  but once again there exist systems that are \emph{energy critical} at $(\zeta, \beta)$ in which case for $\tilde\beta \leq \beta$
  \[ \Tr{ s_\ast(\tilde\beta(-\kappa^2 \Delta + v(\hat x) - \zeta))} > -\infty \iff \tilde\beta = \beta \, . \]

  A key objective of this paper is to prove Theorems 1--3 without assuming mass non-criticality or energy non-criticality.
  In this setting, we are unable to show the existence (or lack thereof) of states that minimize the grand potential but we do show that
  along a minimizing sequence $\kappa^3 w \ast \den \gamma_n \to \mu^\rHF_\kappa$ for some unique $\mu^\rHF_{\kappa}$.

  \section{The Thomas--Fermi Grand Potential Functional}

  In this section we establish some facts about the (positive temperature) Thomas--Fermi grand potential functional.
  To begin, we first prove a useful explicit representation for the functional.

  \begin{proposition}
    \label{omegarhoexplicit}

    Let $\rho \in L^1(\R^3)$, $v : \R^3 \to \R$ be measurable, $w \in L^\infty(\R^3)$, and $s$ be an entropy function.
    Suppose that
    \[ \frac{1}{(2 \pi)^3} \iint s_\ast(\abs{p}^2 + v(x)) \, \dd p \, \dd x > -\infty \, , \]
    and define $f : \R \to (-\infty, \infty]$ by
    \[ f(h) = \frac{1}{(2\pi)^3} \int -s_\ast(\abs{p}^2 - h) \, \dd p \, . \]
    Then $f$ is convex, $x \mapsto f(-v(x))$ is integrable, $v(x) \rho(x) + f^*(\rho(x)) \geq -f(-v(x))$, and
    \begin{align}
      \Omega^\TF(\rho) &= \frac{1}{2} \inner{\rho, \rho}_w + \sup_{\mu \in L^\infty(\R^3)} \left[ \frac{1}{(2\pi)^3} \iint s_\ast(\abs{p}^2 + v(x) + \mu(x)) \, \dd p \, \dd x - \inner{\mu, \rho} \right] \\
      &= \frac{1}{2} \inner{\rho, \rho}_w + \int v(x) \rho(x) + f^\ast(\rho(x)) \, \dd x \, ,
    \end{align}
    where $\rho \mapsto \Omega^\TF(\rho)$ is defined as in \eqref{ftfrhodef} with $\beta = 1$ and $\zeta = 0$.
  \end{proposition}

  \begin{proof}
    Since $\frac{1}{2} \inner{\rho, \rho}_w$ appears on both sides of the desired identity, we can take $w = 0$ without loss of generality.

    For simplicity of notation, we write $\Omega^\TF(\rho, \mu) = \int \phi(x, -\mu(x)) \, \dd x$ with $\phi : \R^3 \times \R \to [-\infty, \infty)$ such that
    \[ \phi(x, \mu) = \mu \rho(x) - \frac{1}{(2 \pi)^3} \int -s_\ast(\abs{p}^2 + v(x) - \mu) \, \dd p \, . \]

    Notice that
    \[ \phi_\infty(x) \equiv \sup_{\mu \in \R} \phi(x, \mu) = \sup_{\mu \in \R} \phi(x, v(x) + \mu) = v(x) \rho(x) + f^\ast(\rho(x)) \, . \]

    With this notation, our desired identity is
    \[ \sup_{\mu \in L^\infty(\R^3)} \int \phi(x, \mu(x)) \, \dd x = \int \phi_\infty(x) \, \dd x \, . \]

    We note that $\int \phi_\infty(x) \, \dd x$ is well-defined as, by definition, $\phi_\infty(x) \geq \phi_0(x) \equiv \frac{1}{(2 \pi)^3} \int s_\ast(\abs{p}^2 + v(x)) \, \dd p$
    and this is an integrable lower bound by the compatibility condition between $v$ and $s$.

    It is clear from the definition of $\phi_\infty$ that
    \[ \sup_{\mu \in L^\infty(\R^3)} \int \phi(x, \mu(x)) \, \dd x \leq \int \phi_\infty(x) \, \dd x \]
    and therefore to conclude we show the other inequality.

    To that end, for each $n \in \N$, let
    \[ \phi_n(x) = \max_{-n \leq \mu \leq n} \phi(x, \mu) \, . \]
    This is well-defined as for each fixed $x$, $\phi(x, -)$ is concave and upper semicontinuous (by Fatou's lemma applied to $f$)
    and therefore the extreme value theorem guarantees the existence of maximizers.
    As such, there exists $\mu_n \in L^\infty(\R^3)$ with $\norm{\mu_n}_{L^\infty(\R^3)} \leq n$ and $\phi_n(x) = \phi(x, \mu_n(x))$.
    To find such $\mu_n$, we simply take a measurable selection such that $\mu_n(x) \in \argmax_{-n \leq \mu \leq n} \phi(x, \mu)$ for each $x \in \R^3$.
    Indeed, such a measurable selection exists because $\phi$ is jointly measurable and upper semicontinuous in $\mu$.

    We then have that $\phi_n(x) \upuparrows \phi_\infty(x)$ for each $x \in \R^3$ while $\phi_0$ is a uniform, integrable, lower bound for the sequence.
    By the monotone convergence theorem (for functions with an integrable lower bound), we then conclude that
    \[ \int \phi_\infty(x) \, \dd x = \lim_{n \to \infty} \int \phi_n(x) \, \dd x = \lim_{n \to \infty} \int \phi(x, \mu_n(x)) \, \dd x \leq \sup_{\mu \in L^\infty(\R^3)} \int \phi(x, \mu(x)) \, \dd x \, . \]
  \end{proof}

  \noindent \emph{Remark.} The above proof shows that
  \[ \int v(x) \rho(x) + f^\ast(\rho(x)) \, \dd x \geq \frac{1}{(2\pi)^3} \iint s_\ast(\abs{p}^2 + v(x)) \, \dd p \dd x \, . \]
  However one cannot in general split this integral without further assumptions on $\rho$.

  Next, we show that $\Omega^\TF(\rho) = \infty$ whenever $\rho < 0$ outside a null set.

  \begin{lemma}
    \label{rhotfpositive}
    Let $v : \R^3 \to \R$ be measurable, $w \in L^\infty(\R^3)$ be repulsive and of positive type, and $s : [0, 1] \to \R$ be an entropy function
    such that the compatibility condition
    \[ \iint s_\ast(\abs{p}^2 + v(x)) \, \dd p \, \dd x > -\infty \]
    holds.
    Define $\Omega^\TF(\rho)$
    as in \eqref{ftfrhodef} with $\beta = 1$ and $\zeta = 0$.

    Then $\Omega^\TF(\rho) = \infty$ unless $\rho(x) \geq 0$ for almost every $x \in \R^3$.
  \end{lemma}

  \begin{proof}
    Let $\rho \in L^1(\R^3)$, $U = \{ \rho(x) \leq 0 \}$ and $\mu_n = n \mathds{1}_U$.
    Then it follows that
    \[ \Omega^\TF(\rho) \geq \Omega^\TF(\rho, \mu_n) = \frac{1}{(2\pi)^3} \iint s_\ast(\abs{p}^2 + v(x) + \mu_n(x)) \, \dd p \, \dd x + \frac{1}{2} \inner{\rho, \rho}_w - \inner{\mu_n, \rho} \, . \]

    Since $s_\ast$ is non-decreasing, $\mu_n \geq 0$ and $\left. \rho \right|_U \leq 0$, it then follows that
    \[ \Omega^\TF(\rho) \geq \frac{1}{(2\pi)^3} \iint s_\ast(\abs{p}^2 + v(x)) \, \dd p \, \dd x + \frac{1}{2} \inner{\rho, \rho}_w + n \norm{\rho}_{L^1(U)} \, . \]

    If $\norm{\rho}_{L^1(U)} \neq 0$, we have $\Omega^\TF(\rho) = \infty$ by taking $n \to \infty$. Otherwise, it follows that $\rho(x) = 0$ for almost every $x \in U$ and thus the set on which $\rho(x) < 0$ has measure zero.
  \end{proof}

  Next, we show a key technical proposition that is the cornerstone of the proof of Theorem \ref{mainthm}.

  \begin{proposition}
    \label{swaptf}
    Let $v : \R^3 \to \R$ be measurable, $w \in L^\infty(\R^3)$ be repulsive and of positive type, and $s : [0, 1] \to \R$ be an entropy function
    such that the compatibility condition
    \[ \iint s_\ast(\abs{p}^2 + v(x) - \zeta) \, \dd p \, \dd x > -\infty \]
    holds.

    Define $\Omega^\TF$ and $\Omega^\TF(-, -)$
    as in \eqref{ftfdef} and \eqref{ftfrhomudef} respectively with $\beta = 1$ and $\zeta = 0$.
    For $\mu \in L^\infty(\R^3)$, let
    $\Omega^\TF(\mu) = \inf_{\rho \in L^1(\R^3)} \Omega^\TF(\rho, \mu)$.

    Then
    \[ \Omega^\TF = \sup_{\mu \in L^\infty(\R^3)} \Omega^\TF(\mu) \]
    and the supremum is attained, i.e., there exists $\mu^\TF \in L^\infty(\R^3)$ with $\mu^\TF \geq 0$ such that $\Omega^\TF = \Omega^\TF(\mu^\TF)$.
  \end{proposition}

  \begin{proof}
    Let $X = L^1(\R^3)$ so that $X^* = L^\infty(\R^3)$, let $f^* : X^* \to (-\infty, \infty]$ such that
    \[ f^*(\mu) = -\frac{1}{(2\pi)^3} \iint s_\ast(\abs{p}^2 + v(x) - \mu(x)) \, \dd p \, \dd x \, , \]
    and let $g : X \to \R$ such that $g(\rho) = \frac{1}{2} \inner{\rho, \rho}_w$.

    To begin, we show that $f^*$ is convex, proper, and weak-star lower semicontinuous and to that end, we write
    \[ f^*(\mu) = \int \phi(x, \mu(x)) \, \dd x \]
    with
    \[ \phi(x, \lambda) = \frac{1}{(2\pi)^3} \int -s_\ast(\abs{p}^2 + v(x) - \lambda) \, \dd p \, . \]

    Since $s_\ast$ is non-positive, $\phi(x, -)$ and therefore $f^*$ are non-negative, while by hypothesis, $f^*(0) < \infty$ so $f^*$ is proper.
    Since $s_\ast$ is concave, for each fixed $x$, $\phi(x, -)$ is an integral of convex functions and thus convex in $\lambda$.
    We then conclude that $f^*$ is convex.
    By Fatou's lemma, $\phi(x, -)$ is also lower semicontinuous for each fixed $x$.

    It follows that $\phi$ satisfies the hypotheses of Ref. \onlinecite[Theorem 6.56]{FonsecaLeoni}: it is measurable,
    lower semicontinuous and convex in the second variable, and is bounded below by $0$.

    We therefore conclude that $f^*$ is sequentially weak-star lower semicontinuous.
    Furthermore, since $f^*$ is convex, the sublevel sets of $f^*$ are convex,
    sequentially weak-star closed and therefore weak-star closed by Ref. \onlinecite[Corollary 12.7]{Conway}.
    It follows that $f^*$ is weak-star lower semicontinuous.

    Recall that the proper, convex, weak-star lower semicontinuous functions are in isomorphism
    with the proper, convex, (norm) lower semicontinuous functions (cf. Ref. \onlinecite[Theorem 2.3.3]{Zalinescu}) and
    therefore (as the notation suggests) $f^*$ is the convex conjugate of some proper,
    convex and lower semicontinuous function $f$ such that $f = f^{**}$.

    It then follows that
    \[ \Omega^\TF(\rho, \mu) =  g(\rho) - f^\ast(-\mu) - \inner{\mu, \rho}\, , \]
    from which we find $\Omega^\TF(\mu) = -g^\ast(\mu) - f^\ast(-\mu)$ and $\Omega^\TF(\rho) = g(\rho) + f^{**}(\rho) = g(\rho) + f(\rho)$.

    Next, we show that $g$ is convex and continuous.
    To that end, by Young's convolution inequality, we have $\inner{\rho, \eta}_w \leq \norm{w}_{L^\infty(\R^3)} \norm{\rho}_{L^1(\R^3)} \norm{\eta}_{L^1(\R^3)}$ and
    since every bounded bilinear form is continuous, it follows that $g$ is continuous.
    On the other hand, $g$ is convex since $w$ is of positive type and thus $\inner{-, -}_w$ is a positive definite bilinear form.

    To conclude, we apply the Fenchel--Rockafellar duality formula (cf. Ref. \onlinecite[Corollary 2.8.5]{Zalinescu}), which applies as $g$ is continuous and $f$ is proper,
    and states
    \begin{align}
      \Omega^\TF = \inf_{\rho \in X} \Omega^\TF(\rho) = \inf_{\rho \in X} f(\rho) + g(\rho) = \sup_{\mu \in X^*} -f^*(-\mu) - g^*(\mu) = \sup_{\mu \in X^*} \Omega^\TF(\mu)
    \end{align}
    and that the supremum is attained.

    Lastly, we show $\mu^\TF \geq 0$.
    To that end, let $(\rho_n)$ be a minimizing sequence so that $\Omega^\TF(\rho_n) \xrightarrow{n \to \infty} \Omega^\TF = \Omega^\TF(\mu^\TF)$,
    let $t \neq 0$ and let $\eta \in L^1(\R^3)$ be arbitrary.
    By Lemma \ref{rhotfpositive}, we can take $\rho_n \geq 0$ without loss of generality.

    We then have
    \[ \Omega^\TF = \Omega^\TF(\mu^\TF) \leq \Omega^\TF(\rho_n, \mu^\TF) \leq \Omega^\TF(\rho_n) \xrightarrow{n \to \infty} \Omega^\TF \, , \]
    and therefore $(\rho_n)$ is also a minimizing sequence for $\Omega^\TF(\mu^\TF)$.

    For each $n$, we have
    \[ \Omega^\TF(\rho_n + t \eta, \mu^\TF) - \Omega^\TF(\rho_n, \mu^\TF) = t \inner{\eta, w \ast \rho_n - \mu^\TF} + \frac{t^2}{2} \inner{ \eta, \eta }_w \, . \]

    The left-hand side is bounded below by $\Omega^\TF(\mu^\TF) - \Omega^\TF(\rho_n, \mu^\TF)$, which goes to $0$ as $n \to \infty$ and thus
    \[ 0 \leq \frac{t^2}{2} \inner{ \eta, \eta }_w + \liminf_{n \to \infty}{\left( t \times \inner{\eta, w \ast \rho_n - \mu^\TF}\right)} \, . \]

    Taking $t \to 0$ from above and below yields
    \[ 0 = \lim_{n \to \infty} {\inner{\eta, w \ast \rho_n - \mu^\TF} } \]
    and thus $w \ast \rho_n \xrightharpoonup{*} \mu^\TF$.
    We conclude, as $w$ is repulsive and $\rho_n \geq 0$, that $\mu^\TF \geq 0$ as the positive functions in $L^\infty(\R^3)$ are weak-star closed.
  \end{proof}

  Indeed, the above shows that the Thomas--Fermi grand potential functional admits a convenient dual functional which can also be used to calculate the Thomas--Fermi grand potential.
  The following results show that the (finite) support of the dual functional is restricted to very nice mean-field potentials.

  \begin{lemma}
    \label{infizernice}
    Let $\mu \in L^\infty(\R^3)$ and let $w \in \mathcal{S}(\R^3)$ be even and of positive type.
    Then
    \[ \label{infval} \inf_{\rho \in L^1(\R^3)} \left[ \frac{1}{2} \inner{\rho, \rho}_w - \inner{\mu, \rho} \right] = -\frac{1}{2} \int \abs{\hat \mu(k)}^2 \, \hat w(k)^{-1} \, \dd k \]
    if $\mu \in L^2(\R^3)$, and
    \[ \inf_{\rho \in L^1(\R^3)} \left[ \frac{1}{2} \inner{\rho, \rho}_w - \inner{\mu, \rho} \right] = -\infty \]
    otherwise.
  \end{lemma}

  \begin{proof}
    To begin, we introduce the sesquilinear inner product
    \[ (\rho, \rho)_w = \iint \rho(x) w(x - y) \bar\rho(y) \, \dd x \, \dd y \, . \]
    and duality pairing
    \[ (\mu, \rho) = \int \bar \mu(x) \rho(x) \, \dd x \, . \]

    We briefly show that
    \[ \inf_{\rho \in L^1(\R^3)} \left[ \frac{1}{2} \inner{\rho, \rho}_w - \inner{\mu, \rho} \right] = \inf_{\rho \in L^1(\R^3; \C)} \left[ \frac{1}{2} (\rho, \rho)_w - \operatorname{Re}{(\mu, \rho)} \right] \, . \]
    Indeed, since $\mu$ is real-valued, we have
    \[ \frac{1}{2} (\rho, \rho)_w - \operatorname{Re}{(\mu, \rho)} = \frac{1}{2} \inner{ \Re{\rho}, \Re{\rho} }_w + \frac{1}{2} \inner{ \Im{\rho}, \Im{\rho} }_w - \inner{\mu, \Re{\rho}} \]
    and therefore adding an imaginary part to a real-valued $\rho$ can only increase the quantity being minimized.
    We therefore work with the complexified pairings for the rest of the proof.

    We first show that if
    \[ \inf_{\rho \in L^1(\R^3)} \left[ \frac{1}{2} \inner{\rho, \rho}_w - \inner{\mu, \rho} \right] > -\infty \]
    then $\mu \in L^2(\R^3)$.

    To that end, since $\mu \in L^\infty(\R^3) \subseteq \mathcal{S}'(\R^3)$ consider $\hat \mu \in \mathcal{S}'(\R^3; \C)$ and let $\phi \in \mathcal{S}(\R^3; \C)$ be arbitrary.
    Let $z \in \C$ such that $z \times (\hat \mu, \phi) = \abs{(\hat \mu, \phi)}$ and let $\alpha \in \R$ be arbitrary.
    It then follows that
    \[ \frac{1}{2} (\alpha z \check \phi, \alpha z \check \phi)_w - \Re{(\mu, \alpha z \check \phi)} = \frac{\alpha^2}{2} \int \abs{\phi(k)}^2 \, \hat w(k) \, \dd k - \alpha \abs{(\hat \mu, \phi)} \, . \]

    Optimizing over $\alpha \in \R$, we find that
    \[ \inf_{\rho \in L^1(\R^3; \C)} \left[ \frac{1}{2} (\rho, \rho)_w - \operatorname{Re}{(\mu, \rho)} \right] \leq -\frac{1}{2} \frac{\abs{ \inner{ \hat \mu, \phi } }^2}{\int \abs{ \phi(k) }^2 \, \hat w(k) \, \mathrm{d} k} \, . \]

    By assumption, and by testing with $\rho = 0$, the left-hand side is finite and non-positive thus
    \begin{align}
      \abs{ \inner{ \hat \mu, \phi } }^2 &\leq -2 \inf_{\rho \in L^1(\R^3; \C)} \left[ \frac{1}{2} (\rho, \rho)_w - \operatorname{Re}{(\mu, \rho)} \right] \left( \int \abs{ \phi(k) }^2 \, \hat w(k) \, \mathrm{d} k \right) \\
      &\leq -2 \inf_{\rho \in L^1(\R^3; \C)} \left[ \frac{1}{2} (\rho, \rho)_w - \operatorname{Re}{(\mu, \rho)} \right] \norm{\hat w}_{L^\infty(\R^3; \C)} \norm{\phi}_{L^2(\R^3; \C)}^2 \, .
    \end{align}

    As Schwartz functions are dense in $L^2(\R^3; \C)$ we conclude by the Riesz representation theorem that $\hat \mu \in L^2(\R^3; \C)$ and thus $\mu \in L^2(\R^3)$.
    In addition, Schwartz functions are dense in the weighted space $L^2(\R^3; \C; \hat w(k) \, \dd k)$ and we similarly conclude that
    \[ \inner{\hat \mu, \phi} = \int \bar\eta(k) \phi(k) \, \hat w(k) \, \dd k \]
    for some $\eta \in L^2(\R^3; \C; \hat w(k) \, \dd k)$
    with the estimate
    \[  \int \abs{\eta(k)}^2 \, \hat w(k) \, \dd k \leq - 2 \inf_{\rho \in L^1(\R^3; \C)} \left[ \frac{1}{2} (\rho, \rho)_w - \operatorname{Re}{(\mu, \rho)} \right] = - 2 \inf_{\rho \in L^1(\R^3)} \left[ \frac{1}{2} \inner{\rho, \rho}_w - \inner{\mu, \rho} \right] \, . \]
    We then find that $\hat \mu(k) = \hat w(k) \eta(k)$ and since $\eta \in L^2(\R^3; \C; \hat w(k) \, \dd k)$
    we conclude
    \[ \int \abs{\eta(k)}^2 \, \hat w(k) \, \dd k = \int \abs{\hat \mu(k)}^2 \, \hat w(k)^{-1} \, \dd k \leq - 2 \inf_{\rho \in L^1(\R^3)} \left[ \frac{1}{2} \inner{\rho, \rho}_w - \inner{\mu, \rho} \right] \, . \]

    Applying the contrapositive, it remains only to show \eqref{infval} in the case where
    \[ \int \abs{\hat \mu(k)}^2 \, \hat w(k)^{-1} \, \dd k < \infty \, . \]

    We conclude by showing \eqref{infval}.
    In this case, since $\mu \in L^2(\R^3)$, we have by completing the square that
    \[ \frac{1}{2} (\check \phi, \check \phi)_w - \Re{(\mu, \check \phi)} + \frac{1}{2} \int \abs{\hat \mu(k)}^2 \, \hat w(k)^{-1} \, \dd k = \frac{1}{2} \int \abs*{\phi(k) - \hat \mu(k) \hat w(k)^{-1} }^2 \, \hat w(k) \, \dd k  \geq 0 \,  . \]
    Since Schwartz functions are contained in $L^1(\R^3; \C)$ and simultaneously dense in $L^2(\R^3; \C; \hat w(k) \, \dd k)$ we conclude that
    \[ 0 \leq \inf_{\rho \in L^1(\R^3; \C)} \left[ \frac{1}{2} (\rho, \rho)_w - \Re{(\mu, \rho)} + \frac{1}{2} \int \abs{\hat \mu(k)}^2 \, \hat w(k)^{-1} \, \dd k \right] \leq \lim_{n \to \infty} \left[ \frac{1}{2} \int \abs*{\phi_n(k) - \hat \mu(k) \hat w(k)^{-1} }^2 \, \hat w(k) \, \dd k \right] = 0 \]
    where $(\phi_n)$ is a sequence of Schwartz functions approximating $\hat \mu(k) \hat w(k)^{-1} = \eta(k)$ in $L^2(\R^3; \C; \hat w(k) \, \dd k)$.
  \end{proof}

  \begin{lemma}
    \label{ftfmunice}
    Let $\mu \in L^\infty(\R^3)$, $v : \R^3 \to \R$ be measurable, $w \in \mathcal{S}(\R^3)$ be even and of positive type, and $s : [0, 1] \to \R$ be an entropy function.
    Suppose that $\inf_{\rho \in L^1(\R^3)} \Omega^\TF(\rho, \mu)$ is finite where $\Omega^\TF(-, -)$ is defined as in \eqref{ftfrhomudef}.

    Then $\mu \in H^\sigma(\R^3)$ for all $\sigma \geq 0$ with
    \[ \norm{\mu}_{H^\sigma(\R^3)}^2 \leq - 2 \left( \, \inf_{\rho \in L^1(\R^3)} \Omega^\TF(\rho, \mu) \right) \norm{ { \langle - \rangle }^{2\sigma} \, \hat w }_{L^\infty(\R^3)} \, . \]
  \end{lemma}

  \begin{proof}
    By Lemma \ref{infizernice}, we have
    \[ \inf_{\rho \in L^1(\R^3)} \Omega^\TF(\rho, \mu) = \begin{cases} \frac{1}{(2\pi)^3} \iint s_\ast(\abs{p}^2 + v(x) + \mu(x)) \, \dd x \, \dd p - \frac{1}{2} \int \abs{\hat \mu(k)}^2 \, \hat w(k)^{-1} \, \dd k & \mu \in L^2(\R^3) \\ -\infty & \mathrm{otherwise}  \end{cases} \, . \]
    Suppose that the left-hand side is finite.
    It follows that $\mu \in L^2(\R^3)$ and since $s_\ast \leq 0$,
    \[ \frac{1}{2} \int \abs{\hat \mu(k)}^2 \, \hat w(k)^{-1} \, \dd k \leq - \inf_{\rho \in L^1(\R^3)} \Omega^\TF(\rho, \mu) \, . \]
    The result then follows from $\norm{\mu}^2_{H^\sigma(\R^3)} = \int \langle k \rangle^{2\sigma} \abs{\hat \mu(k)}^{2} \, \dd k$ and H\"older's inequality.
  \end{proof}

  We now turn to the proofs of Theorem \ref{mainthmtwo} and Theorem \ref{mainthmthree}.
  The following lemma establishes that the zero-temperature Thomas--Fermi grand potential functional established by Nguyen in Ref. \onlinecite{Nguyen}
  is pointwise equal to the Thomas--Fermi grand potential functional defined in \eqref{ftfrhodef} for all $\beta > 0$ when $s = 0$.

  \begin{lemma}
    \label{nguyenomegarep}

    Let $\beta > 0$ and $\zeta \in \R$, let $v \in L^{5/2}_\mathrm{loc.}(\R^3)$ be confining and suppose $s = 0$.
    Define
    \[ \omega^\TF(\rho; \zeta) = \int \left[ \frac{3}{5} (6 \pi^2)^{2/3} \rho(x)^{5/3} + \left[ v(x) - \zeta \right] \rho(x) \right] \, \dd x + \frac{1}{2} \inner{\rho, \rho}_w \, , \]
    so that $\omega^\TF(\zeta) = \inf_{\substack{\rho \in L^1(\R^3; \R_+)}} \omega^\TF(\rho; \zeta)$
    as defined in \eqref{nguyenftfdef}.
    Further define
    \[ \Omega^\TF(\rho, \mu; \zeta, \beta) = \frac{\beta^{-1}}{(2\pi)^3} \iint s_\ast(\beta(\abs{p}^2 + v(x) - \zeta + \mu(x))) \, \dd p \, \dd x + \frac{1}{2} \inner{\rho, \rho}_w - \inner{\mu, \rho} \]
    as in \eqref{ftfrhomudef}.

    Then
    \[ \label{rhofin} \omega^\TF(\rho; \zeta) = \sup_{\mu \in L^\infty(\R^3)} \Omega^\TF(\rho, \mu; \zeta, \beta) \]
    for all $\rho \geq 0$, and
    \[ \label{rhoinf} \sup_{\mu \in L^\infty(\R^3)} \Omega^\TF(\rho, \mu; \zeta, \beta) = \infty \]
    otherwise.
    In particular,
    \[ \omega^\TF(\zeta) = \inf_{\substack{\rho \in L^1(\R^3) \\ \rho \geq 0}} \sup_{\mu \in L^\infty(\R^3)} \Omega^\TF(\rho, \mu; \zeta, \beta) = \inf_{\rho \in L^1(\R^3)} \sup_{\mu \in L^\infty(\R^3)} \Omega^\TF(\rho, \mu; \zeta, \beta) \, . \]
  \end{lemma}

  \begin{proof}
    To begin, we apply Proposition \ref{omegarhoexplicit} with $s = 0$ and external potential $v - \zeta$.
    Since $s = 0$, it suffices to take $\beta = 1$ without loss of generality.
    We first verify that
    \[ \frac{1}{(2\pi)^3} \iint s_\ast(\abs{p}^2 + v(x) - \zeta) \, \dd p \, \dd x = -\frac{1}{(2\pi)^3} \iint (\abs{p}^2 + v(x) - \zeta)_- \, \dd p \, \dd x > -\infty \, . \]

    By computation,
    \[ \int [\abs{p}^2 + h]_- \, \dd p = 4 \pi \int_0^\infty [r^2 + h]_- \, r^2 \, \dd r = -4 \pi \int_0^{h_-^{1 / 2}} r^4 - h_- r^2 \, \dd r = \frac{8}{15} \pi h_-^{5 / 2} \]
    and thus
    \[ \frac{1}{(2\pi)^3} \iint [\abs{p}^2 + v(x) - \zeta]_- \, \dd p \, \dd x = \frac{1}{15 \pi^2} \int \left( v(x) - \zeta \right)_-^{5 / 2} \, \dd x \]
    which is finite as $(v - \zeta)_-$ is compactly supported and $v \in L^{5 / 2}_\mathrm{loc.}(\R^3)$.

    Applying the result yields
    \[ \sup_{\mu \in L^\infty(\R^3)} \Omega^\TF(\rho, \mu; \zeta, \beta) = \frac{1}{2} \inner{\rho, \rho}_w + \int [v(x) - \zeta] \rho(x) + f^*(\rho(x)) \, \dd x \]
    with
    \[ \label{omegaexpr} f(h) = \frac{1}{(2\pi)^3} \int (\abs{p}^2 - h)_- \, \dd p = \frac{1}{15\pi^2} h^{5/2} = \frac{1}{6 \pi^2} \times \frac{2}{5} h^{5/2} \]
    if $h \geq 0$ and $0$ otherwise.

    It then follows that
    \[ f^*(r) = \frac{1}{(6 \pi^2)} \times \frac{3}{5} (6 \pi^2 r)^{5/3} = \frac{3}{5} (6 \pi^2)^{2/3} r^{5/3} \]
    if $r \geq 0$ and $f^\ast(r) = \infty$ otherwise.
    Plugging this in, we conclude \eqref{rhofin}.

    If $\rho < 0$ on some set of positive measure, then $f^\ast \circ \rho = \infty$ on a set of positive measure and therefore integrates to $\infty$.
    Otherwise,
    \[ \sup_{\mu \in L^\infty(\R^3)} \Omega^\TF(\rho, \mu; \zeta, \beta) = \frac{1}{2} \inner{\rho, \rho}_w + \int [v(x) - \zeta] \rho(x) + \frac{3}{5} (6 \pi^2)^{2/3} \rho(x)^{5/3} \, \dd x \]
  \end{proof}

  We then proceed to the proof of Theorem \ref{mainthmtwo}.

  \begin{proof}[Proof of Theorem \ref{mainthmtwo}]
    Since $v - \zeta$ satisfies the hypotheses for the theorem whenever $v$ does, we take $\zeta = 0$ without loss of generality and remove $\zeta$ from all notation.
    For $\mu \in L^\infty(\R^3)$, let $\Omega^\TF(\mu; \beta) = \inf_{\rho \in L^1(\R^3)} \Omega^\TF(\rho, \mu; \beta)$ with $s$ and $\beta$ defined as in the statement of the theorem
    and $\omega^\TF(\mu) = \inf_{\rho \in L^1(\R^3)} \Omega^\TF(\rho, \mu; \beta)$ with $s = 0$ and $\beta = 1$.

    By Proposition \ref{swaptf} and Lemma \ref{nguyenomegarep}, we have $\Omega^\TF(\beta) = \Omega^\TF(\mu_\beta; \beta)$ and $\omega^\TF = \omega^\TF(\mu_\infty)$ for some $\mu_\beta$
    and $\mu_\infty$ in $L^\infty(\R^3)$.

    By the same manipulations as in the proof of Theorem \ref{mainthm} we have
    \[ \Omega^\TF(\beta) - \omega^\TF \geq \frac{1}{(2\pi)^3} \iint \delta_\beta(\abs{p}^2 + v(x) + \mu_\infty(x)) \, \dd p \, \dd x \, , \]
    and
    \[ \Omega^\TF(\beta) - \omega^\TF \leq \frac{1}{(2\pi)^3} \iint \delta_\beta(\abs{p}^2 + v(x) + \mu_\beta(x)) \, \dd p \, \dd x \, , \]
    where
    \[ \delta_\beta(h) = \beta^{-1} \delta_{\beta = 1}(\beta h) = \beta^{-1} s_\ast(\beta h) + (h)_- = \begin{cases} \beta^{-1} s_\ast(\beta h) & h \geq 0 \\ \beta^{-1} s_\ast(\beta h) - h & h \leq 0 \end{cases} \, . \]

    Additionally, we have $s_\ast(h) - h = \inf_{0 \leq r \leq 1} h(r - 1) - s(r) = \inf_{0 \leq r \leq 1} (-h) r - s(1 - r) = \tilde s_\ast(-h)$
    where $\tilde s(r) = s(1 - r)$ and thus
    \[ \delta_\beta(h) =  \begin{cases} \beta^{-1} s_\ast(\beta h) & h \geq 0 \\ \beta^{-1} \tilde s_\ast(-\beta h) & h \leq 0 \end{cases} \, . \]

    Using established properties of entropy functions, in particular Lemma \ref{entropydecay}, we conclude that
    $\delta_\beta \leq 0$ and that $\delta_\beta$ is non-decreasing for $h \geq 0$, non-increasing for $h \leq 0$ and therefore has a global minimum at $h = 0$.
    In particular, we have
    \[ 0 \geq \Omega^\TF(\beta) - \omega^\TF \geq \frac{1}{(2\pi)^3} \iint \delta_\beta(\abs{p}^2 + v(x) + \mu_\infty(x)) \, \dd p \, \dd x \, . \]

    To conclude, we recall from \eqref{betamonotoneone}-\eqref{betamonotonetwo} that $\beta \mapsto \beta^{-1} s_\ast(\beta h)$ is non-decreasing
    and therefore by the dominated convergence theorem,
    with the integrand dominated (from below) pointwise by $\delta_{\beta_0}(\abs{p}^2 + v(x) + \mu_\infty(x))$,
    we conclude that $\Omega^\TF(\beta) - \omega^\TF \to 0$ as $\beta \to \infty$.
  \end{proof}

  We now conclude with the proof of Theorem \ref{mainthmthree}.

  \begin{proof}[Proof of Theorem 3]
    Suppose $s(r) = -r\log r - (1 - r) \log{(1-r)}$ and let $\beta > 0$ and $\zeta \in \R$ be arbitrary.
    By scaling, it is equivalent to take $s(r) = \beta^{-1} (-r\log r - (1 - r) \log{(1-r)})$
    and $v(x) \to v(x) - \zeta$.

    To begin, we verify the assumptions of Proposition \ref{omegarhoexplicit}.
    Indeed, by computation, we find
    \[ s_\ast(h) = -\beta^{-1} \log{(1 + \exp(-\beta h))} \, ,  \]
    and therefore $s_\ast'(h) = \left( 1 + \exp(\beta h) \right)^{-1}$.
    In particular, $0 \leq s_\ast'(h) \leq \min{\{ 1,\,\exp(-\beta h) \}}$.

    For each fixed $x$ we have
    \[ 0 \leq - \int s_\ast(\abs{p}^2 + v(x)) \, \dd p = -4\pi \int_0^\infty r^2 s_\ast(r^2 + v(x)) \, \dd r = -2\pi \int_0^\infty u^{1/2} s_\ast(u + v(x)) \, \dd u \, . \]
    After integration by parts, we have
    \begin{align}
      0 \leq - \int s_\ast(\abs{p}^2 + v(x)) \, \dd p &= \frac{4}{3} \pi \int_0^\infty u^{3/2} s_\ast'(u + v(x)) \, \dd u \\
      &\leq \frac{4}{3} \pi \int_0^{v(x)_-} u^{3/2} \, \dd u + \frac{4}{3} \pi \exp(-\beta v(x)) \int_{v(x)_-}^\infty u^{3/2} \exp(-\beta u) \, \dd u \\
      &= \frac{4}{3} \pi \int_0^{v(x)_-} u^{3/2} \, \dd u + \frac{4}{3} \pi \exp(-\beta v(x)_+) \int_{0}^\infty (u + v(x)_-)^{3/2} \exp(-\beta u) \, \dd u \, .
    \end{align}

    Using the bound $(x + y)^{3/2} \leq 2^{1/2} (x^{3/2} + y^{3/2})$ we have
    \[ \int_{0}^\infty (u + v(x)_-)^{3/2} \exp(-\beta u) \, \dd u \leq 2^{1/2} \left( \beta^{-5/2} \Gamma(5/2) + \beta^{-1} v(x)_-^{3/2} \right) \, . \]
    We therefore conclude that
    \[ - \int s_\ast(\abs{p}^2 + v(x)) \, \dd p \leq \frac{8}{15} \pi v(x)_-^{5/2} + \frac{2^{5/2}}{3} \pi \beta^{-1} v(x)_-^{3/2} + \frac{2^{5/2}}{3} \pi \Gamma(5/2) \beta^{-5/2} \exp(-\beta v(x)_+) \, . \]
    All three terms are integrable: the first two from $v_- \in L^{5/2}(\R^3)$ with compact support and the last by assumption.
    It therefore follows that $\iint s_\ast(\abs{p}^2 + v(x)) \, \dd p \, \dd x$ is finite.

    Applying Proposition \ref{omegarhoexplicit}, we then conclude that
    \[ \Omega^\TF(\rho; \zeta, \beta) = \frac{1}{2} \inner{\rho, \rho}_w + \int [v(x) - \zeta] \rho(x) + f^\ast(\rho(x)) \, \dd x \]
    with
    \[ f(h) = \frac{1}{(2\pi)^3} \int -s_\ast( \abs{p}^2 - h ) \, \dd p = \frac{\beta^{-1}}{(2\pi)^3} \int \log\left( 1 + \exp{\left( -\beta (\abs{p}^2 - h)\right)} \right) \, \dd p  \, . \]

    Indeed, a standard calculation for $r > 0$ shows that $f^*(r) = r h(r) - f(h(r))$ with $h$ the inverse function of
    $f' : \R \to (0, \infty)$ where
    \[ f'(h) = \frac{1}{(2\pi)^3} \int s_\ast'(\abs{p}^2 - h) \, \dd p = \frac{1}{(2\pi)^3} \int \left( 1 + \exp{(\beta (\abs{p}^2 - h))} \right)^{-1} \, \dd p \, . \]

    It follows that for $r > 0$, $h(r)$ is the unique real number such that
    \[ r = \frac{1}{(2\pi)^3} \int \left( 1 + \exp{(\beta (\abs{p}^2 - h(r)))} \right)^{-1} \, \dd p \, . \]

    One can also calculate that $f^*(0) = 0$ and $f^*(r) = \infty$ for $r < 0$.
    It remains to justify the restriction of the infimum in \eqref{ftfdef} to only densities $\rho$ that are non-negative almost everywhere,
    but if $\rho < 0$ on a set of positive measure, it follows that $f^\ast \circ \rho$ is equal to $\infty$ on a set of positive measure and
    $\Omega^\TF(\rho; \zeta, \beta) = \infty$.
  \end{proof}

  \section{The Reduced Hartree--Fock Grand Potential Functional}

  To begin this section, we briefly review some well-known technicalities in the definition of the reduced Hartree--Fock grand potential functional.
  Since $H_\kappa$ is unbounded, the quantity
  \[ \Tr \left( -\kappa^2 \Delta + v(\hat x) - \zeta \right) \gamma - \beta^{-1} \Tr s(\gamma) \]
  in the definition of $\Omega^\rHF_\kappa (\gamma; \zeta, \beta)$ given in \eqref{fksdef} is understood in the quadratic-form sense as
  \[ K_{H_\kappa - \zeta \mathds{1}} (\gamma) - \beta^{-1} S(\gamma) \]
  where $S(\gamma) = \Tr s(\gamma)$. More generally, for a semibounded operator $H$ with form domain $Q(H)$, choose $E \in \R$ such that $H + E \mathds{1} \geq 0$ and define
  \[ \label{khdef} K_H(\gamma) = \norm{ (H + E)^{1 / 2} \gamma^{1 / 2} }_{\mathfrak{S}_2(\mathcal{H})}^2 - E \Tr \gamma \in (-\infty, \infty] \]
  when $\Ran \gamma^{1 / 2} \subseteq Q(H)$ and $K_H(\gamma) = \infty$ otherwise.

  This definition does not depend on the choice of $E$. Moreover, whenever the right-hand sides are finite,
  \[ K_H(t \gamma_1 + (1 - t) \gamma_2) = t K_H(\gamma_1) + (1 - t) K_H(\gamma_2) \]
  and $K_{H + A} (\gamma) = K_H(\gamma) + \Tr \left( A \gamma \right)$ for $A$ bounded.

  If $H$ has compact resolvent, with eigenvectors $(u_n)$ and corresponding eigenvalues $(\lambda_n)$, then this agrees with the spectral representation
  \[ \label{eigeniden} K_H(\gamma) = \sum_{n = 1}^\infty \lambda_n \inner{u_n, \gamma u_n} \]
  introduced above. The following lemma shows, in particular, that the restriction $\gamma \in \DO_\kappa{(\mathcal{H})}$ avoids the indeterminate form $\infty - \infty$ in $\Omega_\kappa^\rHF(\gamma; \zeta, \beta)$ and records the convexity used below.

  \begin{lemma}
    \label{fksconvexity}

    Let $H$ be semibounded with compact resolvent and (form) domain $Q(H)$ and let $s$ be an entropy function.
    Suppose that $\Tr s_\ast(H)$ is finite.

    Define $\mathcal{F}_H : \DO_H{(\mathcal{H})} \to \R$ by
    \[ \mathcal{F}_H(\gamma) = K_H(\gamma) - S(\gamma) \]
    where
    \[ \DO_H{(\mathcal{H})} = \{ \gamma \in \DO{(\mathcal{H})} : K_H(\gamma) < \infty \} \, . \]

    Then $\DO_H{(\mathcal{H})}$ is convex and $\mathcal{F}_H$ is well-defined, convex and bounded below by $\Tr s_\ast(H)$.
  \end{lemma}

  \begin{proof}
    Since $K_H$ is affine, $\DO_H{(\mathcal{H})}$ is convex. Moreover, by \eqref{eigeniden},
    \[ \mathcal{F}_H(\gamma) = \sum_{n = 1}^\infty \inner{u_n, [\lambda_n \gamma - s(\gamma) - s_\ast(\lambda_n)] u_n} + \Tr s_\ast(H) \, . \]
    By Lemma \ref{operatorfenchel}, each summand in the infinite sum is non-negative and therefore
    \[ \mathcal{F}_H(\gamma) \geq \Tr s_\ast(H) \, . \]

    The same inequality gives
    \[ \Tr s(\gamma) \leq K_H(\gamma) - \Tr s_\ast(H) \]
    and hence $S(\gamma)$ is finite whenever $\gamma \in \DO_H(\mathcal{H})$. Thus $\mathcal{F}_H$ is well-defined. Finally, $S$ is concave (cf. Ref. \onlinecite{LiebPedersen}) while $K_H$ is affine, and therefore $\mathcal{F}_H$ is convex.
  \end{proof}

  With these technicalities explained, we show that the reduced Hartree--Fock grand potential can be
  calculated by means of the optimization of an appropriate dual functional.

  Indeed, we have the following:

  \begin{proposition}
  \label{ksdftequiv}
    Let $\kappa > 0$, $v \in L^{5/2}_\mathrm{loc.}(\R^3)$ be confining, $w \in L^\infty(\R^3)$ be even, repulsive and of positive type, and $s : [0, 1] \to \R$ be an entropy function
    such that $s_\ast \in C^2(\R)$ and the compatibility condition
    \[ \iint s_\ast(\abs{p}^2 + v(x)) \, \dd x \, \dd p > -\infty \]
    holds.
    Let
    \[ \Omega_\kappa^\DFT(\rho, \mu) = \Tr s_\ast(H_\kappa + \mu(\hat x)) + \frac{\kappa^3}{2} \inner{ \rho, \rho }_w - \inner{ \mu, \rho } \, , \]
    $\Omega_\kappa^\DFT(\mu) = \inf_{\rho \in L^1(\R^3)} \Omega_\kappa^\DFT(\rho, \mu)$, and let $\Omega^\rHF_\kappa$ and $\gamma \mapsto \Omega^\rHF_\kappa(\gamma)$ be defined as in \eqref{fksdef} with $\beta = 1$ and $\zeta = 0$.

    Then there exists $\mu_\kappa^\rHF \in L^\infty(\R^3)$ such that the following hold:
    \begin{enumerate}
      \item For every minimizing sequence $(\gamma_n)$ of $\Omega_\kappa^\rHF(-)$, i.e., $\Omega_\kappa^\rHF(\gamma_n) \xrightarrow{n \to \infty} \Omega_\kappa^\rHF$, we have $\kappa^3 w \ast \den \gamma_n \to \mu_\kappa^\rHF$ in $L^\infty(\R^3)$.
      \item $\Omega_\kappa^\rHF = \sup_{\mu \in L^\infty(\R^3)} \Omega_\kappa^\DFT(\mu) = \Omega_\kappa^\DFT(\mu^\rHF_\kappa)$.
    \end{enumerate}
  \end{proposition}

  Before proceeding to the proof, we show the following lemma, which
  establishes that a class of approximate Gibbs states belongs to the domain of $\gamma \mapsto \Omega^\rHF(\gamma)$.

  \begin{lemma}
    \label{etaalphaepsilon}

    Let $H$ be semibounded with compact resolvent and (form) domain $Q(H)$, let $A \in \mathcal{B}(\mathcal{H})$ such that $A \geq 0$,
    let $s : [0, 1] \to \R$ be an entropy function such that $s_\ast \in C^1(\R)$ and let
    \[ \eta_{\alpha, \epsilon} = s_\ast'(\alpha (H + A + \epsilon)) \, . \]

    Then for each $\alpha > 1$ and $\epsilon > 0$, we have $0 \leq \eta_{\alpha, \epsilon} \leq \mathds{1}$ and $\operatorname{Ran} \eta_{\alpha, \epsilon}^{1/2} \subseteq Q(H)$
    with the estimate
    \[ \Tr \eta_{\alpha, \epsilon} \leq - \frac{1}{\epsilon} \Tr s_\ast(H + A) \leq - \frac{1}{\epsilon} \Tr s_\ast(H) \, , \]
    and $\eta_{\alpha, \epsilon} \in \DO_H(\mathcal{H})$ whenever $\Tr s_\ast(H) < \infty$ and it satisfies
    \[ K_{H} (\eta_{\alpha, \epsilon}) \leq K_{H + A} (\eta_{\alpha, \epsilon}) \leq - \frac{1}{\alpha - 1} \Tr s_\ast(H + A) \leq - \frac{1}{\alpha - 1} \Tr s_\ast(H) \, . \]
  \end{lemma}

  \begin{proof}
    To begin, we have $0 \leq \eta_{\alpha, \epsilon} \leq \mathds{1}$ since $0 \leq s_\ast' \leq 1$.
    Next, we show the trace estimate.
    By concavity, and since $s_\ast$ is non-positive, we have
    \[ 0 \leq \alpha \epsilon s_\ast'(\alpha (\lambda + \epsilon)) \leq s_\ast(\alpha (\lambda + \epsilon)) - s_\ast(\alpha \lambda) \leq -s_\ast(\alpha \lambda) \, . \]
    We also have
    \begin{align}
      \label{betamonotoneone} \frac{\dd}{\dd \alpha} \alpha^{-1} s_\ast(\alpha \lambda) = \frac{\lambda}{\alpha} s_\ast'(\alpha \lambda) - \frac{1}{\alpha^2} s_\ast(\alpha \lambda) &= \frac{1}{\alpha^2} \left[ \lambda \alpha s_\ast'(\alpha \lambda) - s_\ast(\alpha \lambda) \right] \\
      \label{betamonotonetwo} &= \frac{1}{\alpha^2} \left[ s \circ s_\ast'(\alpha \lambda) \right] \geq 0
    \end{align}
    and therefore since $\alpha > 1$
    \[ 0 \leq \epsilon s_\ast'(\alpha(\lambda + \epsilon)) \leq -\alpha^{-1} s_\ast(\alpha \lambda) \leq -s_\ast(\lambda) \, . \]

    Summing over the eigenvalues of $H + A$, we then find $\epsilon \Tr \eta_{\alpha, \epsilon} \leq -\Tr s_\ast(H + A)$
    and $-\Tr s_\ast(H + A) \leq -\Tr s_\ast(H)$ follows from the min-max principle, the positive semi-definiteness of $A$ and the monotonicity of $s_\ast$.

    Having shown the trace estimate, we move on to the inclusion $\Ran \eta_{\alpha, \epsilon}^{1/2} \subseteq Q(H)$.
    Again, by concavity, we have
    \[ s_\ast(\lambda + \epsilon) \leq s_\ast(\alpha (\lambda + \epsilon)) + s_\ast'(\alpha (\lambda + \epsilon)) \times (\lambda + \epsilon) \cdot (1 - \alpha) \]
    and since $\alpha > 1$ this implies
    \begin{align}
      \lambda s_\ast'(\alpha(\lambda + \epsilon)) &\leq \frac{s_\ast(\alpha(\lambda + \epsilon)) - s_\ast(\lambda + \epsilon)}{\alpha - 1} - \epsilon s_\ast'(\alpha(\lambda + \epsilon)) \\
      &\leq -\frac{s_\ast(\lambda + \epsilon)}{\alpha - 1}
    \end{align}
    for all $\lambda \in \R$ and in particular, for $E \in \R$ such that $\lambda + E \geq 0$ we have
    \[ 0 \leq (\lambda + E) \, s_\ast'(\alpha(\lambda + \epsilon)) \leq -\frac{s_\ast(-E)}{\alpha - 1} + E_+ \, . \]

    Since $Q(H + A) = Q(H)$, it suffices to show $\Ran \eta_{\alpha, \epsilon}^{1/2} \subseteq Q(H + A)$ or equivalently to show the case $A = 0$.
    Choose $E$ so that $H + E \mathds{1} > 0$.
    The above inequality then implies that
    \[ 0 \leq \eta_{\alpha, \epsilon} \lesssim_{E, \alpha} (H + E)^{-1} \]
    and therefore by Douglas' lemma we have $\Ran \eta_{\alpha, \epsilon}^{1 / 2} \subseteq \Ran{(H + E)^{-1/2}}$.
    Since $Q(H) = \operatorname{Ran}{(H + E)^{-1/2}}$ for all $E \in \R$ such that $H + E \mathds{1} > 0$ we are done.

    To conclude, we show the energy inequality.
    Returning to the $A \geq 0$ setting, it follows from the estimates above that
    \begin{align}
      K_{H + A}(\eta_{\alpha, \epsilon}) \leq -\frac{1}{\alpha - 1} \Tr s_\ast(H + A + \epsilon) &\leq -\frac{1}{\alpha - 1} \Tr s_\ast(H + A) \leq -\frac{1}{\alpha - 1} \Tr s_\ast(H)
    \end{align}
    while $K_H(\eta_{\alpha, \epsilon}) \leq K_{H + A}(\eta_{\alpha, \epsilon})$.
  \end{proof}

  We now proceed to the proof of Proposition \ref{ksdftequiv}.

  \begin{proof}[Proof of Proposition \ref{ksdftequiv}]
    Following the proof of Theorem \ref{mainthm}, $\Omega^\rHF_\kappa(\gamma)$ is bounded below.
    As such, let $(\gamma_n)$ be a minimizing sequence and let $\rho_n = \operatorname{den} \gamma_n$.
    By the convexity of $\mathcal{F}_{H_\kappa}$, see Lemma \ref{fksconvexity}, we have
    \[ \Omega_\kappa^\rHF\left( \frac{\gamma_n + \gamma_m}{2} \right) \leq \frac{1}{2} \mathcal{F}_{H_\kappa}(\gamma_n) + \frac{1}{2} \mathcal{F}_{H_\kappa}(\gamma_m) + \frac{\kappa^3}{8} \inner{\rho_n + \rho_m, \rho_n + \rho_m}_w \, .  \]

    By the parallelogram law, we have
    \[ \frac{1}{2} \inner{\rho_n + \rho_m, \rho_n + \rho_m}_w = \inner{\rho_n, \rho_n}_w + \inner{\rho_m, \rho_m}_w - \frac{1}{2} \inner{\rho_n - \rho_m, \rho_n - \rho_m}_w \, , \]
    and thus
    \[ \Omega_\kappa^\rHF \leq \Omega_\kappa^\rHF\left( \frac{\gamma_n + \gamma_m}{2} \right) \leq \frac{1}{2} \Omega_\kappa^\rHF(\gamma_n) + \frac{1}{2} \Omega_\kappa^\rHF(\gamma_m) - \frac{\kappa^3}{8} \inner{\rho_n - \rho_m, \rho_n - \rho_m}_w \, .  \]

    We therefore conclude that
    \[ \frac{\kappa^3}{8} \norm{\rho_n - \rho_m}_w^2 \leq \frac{1}{2} \left[ \Omega_\kappa^\rHF(\gamma_n) - \Omega_\kappa^\rHF \right] + \frac{1}{2} \left[ \Omega_\kappa^\rHF(\gamma_m) - \Omega_\kappa^\rHF \right] \, . \]
    It follows that $(\rho_n)$ is Cauchy with respect to $\inner{-, -}_w$ and has a limit $\rho_\kappa^\rHF \in \mathfrak{H}_w$, where $\mathfrak{H}_w$ is the completion of $L^1(\R^3)$
    with respect to $\inner{-, -}_w$.

    This limit is also unique, as if two minimizing sequences $(\gamma_n)$ and $(\tilde \gamma_n)$ had $\rho_n \to \rho_\kappa^\rHF$
    and $\tilde \rho_n \to \tilde\rho_\kappa^\rHF$, the intertwined sequence $(\gamma_1, \tilde \gamma_1, \gamma_2, \tilde\gamma_2, \ldots)$ would again be a minimizing sequence
    but the corresponding sequence $(\rho_1, \tilde \rho_1, \ldots)$ could not converge unless $\tilde\rho_\kappa^\rHF = \rho_\kappa^\rHF$.

    Note that by Young's convolution inequality, $\rho \mapsto w \ast \rho$ is a well-defined map from $L^1(\R^3)$ to $L^\infty(\R^3)$.
    Furthermore,
    \[ \inner{\eta, w \ast \rho} = \inner{\eta, \rho}_w \leq \norm{\eta}_w \norm{\rho}_w \leq \norm{w}_{L^\infty(\R^3)}^{1 / 2} \norm{\rho}_w \]
    whenever $\norm{\eta}_{L^1(\R^3)} = 1$ and thus $\rho \mapsto w \ast \rho$ extends uniquely to a bounded linear operator from $\mathfrak{H}_w \to L^\infty(\R^3)$.
    This then concludes the first part of the proof as bounded linear operators are continuous and therefore
    for any minimizing sequence $(\gamma_n)$, $\kappa^3 w \ast \den \gamma_n \xrightarrow{n \to \infty} \kappa^3 w \ast \rho_\kappa^\rHF \equiv \mu_\kappa^\rHF$.

    To conclude the proof, we first show that along any minimizing sequence $(\gamma_n)$, we have
    \[ \Omega_\kappa^\rHF(\gamma_n) - \Omega_\kappa^\DFT(\mu_n) \to 0 \, , \]
    where $\mu_n = \kappa^3 w \ast \operatorname{den} \gamma_n$.
    To that end, we have
    \begin{align}
      \Omega_\kappa^\rHF(\gamma_n) - \Omega_\kappa^\DFT(\mu_n) &= \Omega_\kappa^\rHF(\gamma_n) - \Omega_\kappa^\DFT(\rho_n, \mu_n) \\
      &=  K_{H_\kappa}(\gamma_n) - S(\gamma_n) - \Tr s_\ast(H_\kappa + \mu_n(\hat x)) + \inner{\rho_n, \mu_n} \\
      &= K_{H_\kappa + \mu_n(\hat x)}(\gamma_n) - S(\gamma_n) - \Tr s_\ast(H_\kappa + \mu_n(\hat x)) \\
      &\equiv \Delta_{H_\kappa + \mu_n(\hat x)}(\gamma_n) \, ,
    \end{align}
    which is formally the Fenchel defect of $S$ at $(H_\kappa + \mu_n(\hat x), \gamma_n)$.
    We now show that $\Delta_{H_\kappa + \mu_n(\hat x)}(\gamma_n) \xrightarrow{n \to \infty} 0$.

    Let $(\eta_n) \subseteq \DO_\kappa{(\mathcal{H})}$ be a sequence we will choose later.
    It follows for $t \in (0, 1]$ and $\gamma_{n, t} = t \eta_n + (1 - t) \gamma_n = \gamma_n + t (\eta_n - \gamma_n)$ that
    \begin{align}
      \Omega_\kappa^\rHF(\gamma_{t, n}) - \Omega_\kappa^\rHF(\gamma_n) &= \mathcal{F}_{H_\kappa}(\gamma_{t, n}) - \mathcal{F}_{H_\kappa}(\gamma_n) + \frac{\kappa^3}{2} \left[ \inner{\rho_{n, t}, \rho_{n, t}}_w - \inner{\rho_n, \rho_n}_w \right] \\
      &= \mathcal{F}_{H_\kappa}(\gamma_{t, n}) - \mathcal{F}_{H_\kappa}(\gamma_n) + \kappa^3 t \inner{\operatorname{den}{\eta_n} - \rho_n, \rho_n} + \frac{\kappa^3}{2} t^2 \norm{\operatorname{den} \eta_n - \rho_n}_w^2 \\
      &\leq t \left[ \mathcal{F}_{H_\kappa}(\eta_n) - \mathcal{F}_{H_\kappa}(\gamma_n) \right] + \kappa^3 t \inner{\operatorname{den}{\eta_n} - \rho_n, \rho_n} + \frac{\kappa^3}{2} t^2 \norm{\operatorname{den} \eta_n - \rho_n}_w^2 \\
      &= t \left[ \mathcal{F}_{H_\kappa + \mu_n(\hat x)}(\eta_n) - \mathcal{F}_{H_\kappa + \mu_n(\hat x)}(\gamma_n) \right] + \frac{\kappa^3}{2} t^2 \norm{\operatorname{den} \eta_n - \rho_n}_w^2 \\
      &= t \left[ \Delta_{H_\kappa + \mu_n(\hat x)}(\eta_n) - \Delta_{H_\kappa + \mu_n(\hat x)}(\gamma_n) \right] + \frac{\kappa^3}{2} t^2 \norm{\operatorname{den} \eta_n - \rho_n}_w^2 \, .
    \end{align}

    Since $\Omega_\kappa^\rHF(\gamma_{t, n}) \geq \Omega_\kappa^\rHF$, and by Lemma \ref{operatorfenchel}, we conclude that
    \[ 0 \leq \Delta_{H_\kappa + \mu_n(\hat x)}(\gamma_n) \leq \Delta_{H_\kappa + \mu_n(\hat x)}(\eta_n) + \frac{\Omega_\kappa^\rHF(\gamma_n) - \Omega_\kappa^\rHF}{t} + \frac{\kappa^3}{2} t \norm{\operatorname{den} \eta_n - \rho_n}_w^2 \, . \]

    An intuitive choice of $\eta_n$ is $\eta_n = s_\ast'(H_\kappa + \mu_n(\hat x))$ in which case $\Delta_{H_\kappa + \mu_n(\hat x)}(\eta_n)$ would in fact vanish
    for all $n \in \mathbb{N}$.
    The issue is that we cannot in general guarantee that $\eta_n \in \DO_\kappa{(\mathcal{H})}$ without further assumptions on the entropy or confining potential $v$.
    Instead, we take $\eta_n = \eta_n^{\alpha, \epsilon} = s_\ast'(\alpha (H_\kappa + \mu_n(\hat x) + \epsilon))$ for $\alpha > 1$ and $\epsilon > 0$.

    Indeed, we show in Lemma \ref{etaalphaepsilon} that $\eta_n^{\alpha, \epsilon} \in \DO_\kappa(\mathcal{H})$ with
    \[ \norm{\operatorname{den} \eta_n^{\alpha, \epsilon}}_{L^1(\R^3)} \leq -\frac{1}{\epsilon} \Tr s_\ast(H_\kappa) \, . \]
    Additionally, since $\gamma_n$ is a minimizing sequence with $\Omega_\kappa^\rHF(\gamma = 0) = 0$, we have $\Omega_\kappa^\rHF(\gamma_n) \leq 1 - \Tr s_\ast(H_\kappa)$
    for sufficiently large $n$.

    In particular,
    \[ \mathcal{F}_{H_\kappa}(\gamma_n) + \frac{\kappa^3}{2} \norm{\rho_n}^2_w \leq 1 - \Tr s_\ast(H_\kappa) \implies \frac{\kappa^3}{2} \norm{\rho_n}^2_w \leq 1 - 2 \Tr s_\ast(H_\kappa) \, . \]

    We conclude that
    \[ \norm{\den \eta_n^{\alpha, \epsilon} - \rho_n}^2_w \lesssim \norm{\den \eta_n^{\alpha, \epsilon}}^2_w + \norm{\rho_n}^2_w \leq \norm{w}_{L^\infty(\R^3)} \norm{\den \eta_n^{\alpha, \epsilon}}^2_{L^1(\R^3)} + \norm{\rho_n}^2_w \lesssim_\epsilon 1 \, . \]

    In particular, by taking $n \to \infty$ and then $t \to 0$ we conclude that
    \[ 0 \leq \limsup_{n \to \infty} \Delta_{H_\kappa + \mu_n(\hat x)}(\gamma_n) \leq \limsup_{n \to \infty} \Delta_{H_\kappa + \mu_n(\hat x)}(\eta_n^{\alpha, \epsilon}) \, . \]

    To conclude, we have that $\Delta_{H_\kappa + \mu_n(\hat x)} (\eta_n^{\alpha, \epsilon})$ is equal to
    \[ \Tr s_\ast(\alpha(H_\kappa + \mu_n(\hat x) + \epsilon)) - \Tr s_\ast(H_\kappa + \mu_n(\hat x)) - (\alpha - 1) K_{H_\kappa + \mu_n(\hat x) + \epsilon}(\eta_n^{\alpha, \epsilon}) - \epsilon \Tr \eta_n^{\alpha, \epsilon} \]

    Since $\eta_n^{\alpha, \epsilon}$ is a function of $H_\kappa + \mu_n(\hat x)$ we can write
    \[ \Delta_{H_\kappa + \mu_n(\hat x)} (\eta_n^{\alpha, \epsilon}) = \Tr r_{\alpha, \epsilon}(H_\kappa + \mu_n(\hat x)) \]
    with
    \[ r_{\alpha, \epsilon}(p) = s_\ast(\alpha (p + \epsilon)) - s_\ast(p) - (\alpha (p + \epsilon) - p) \, s_\ast'(\alpha (p + \epsilon)) \, . \]

    Following the proof of Lemma \ref{lsc}, the map $\mu \mapsto \lambda_n(\mu)$, where $\lambda_n(\mu)$ is the $n$-th smallest eigenvalue of $H_\kappa + \mu(\hat x)$, is continuous.
    Notice that for $p \geq - \frac{\alpha}{\alpha - 1} \epsilon$ we have $0 \leq r_{\alpha, \epsilon}(p) \leq -s_\ast(p)$.
    In particular, if
    \[ \epsilon = \epsilon(\alpha) = (\alpha - 1) \max \left\{ \frac{E_\kappa}{\alpha}, 1 \right\} \, , \]
    we have
    \[ \lambda_k(\mu_n) \geq \lambda_k(0) \geq -E_\kappa \geq - \frac{\alpha}{\alpha - 1} \epsilon(\alpha) \implies 0 \leq r_{\alpha, \epsilon(\alpha)}(\lambda_k(\mu_n)) \leq -s_\ast(\lambda_k(\mu_n)) \leq -s_\ast(\lambda_k(0)) \]
    since $\mu_n \geq 0$.
    It then follows from dominated convergence that
    \[ \lim_{n \to \infty} \Tr r_{\alpha, \epsilon(\alpha)} (H_\kappa + \mu_n(\hat x)) = \Tr r_{\alpha, \epsilon(\alpha)} (H_\kappa + \mu_\kappa^\rHF(\hat x)) \, . \]

    It therefore suffices to show that
    \[ \lim_{\substack{\alpha \to 1}} \Tr r_{\alpha, \epsilon(\alpha)} (H_\kappa + \mu_\kappa^\rHF(\hat x)) = 0 \, , \]
    but we can once again apply dominated convergence to conclude exactly that.

    Next, we show that $\Omega_\kappa^\DFT(\mu_n) \to \Omega_\kappa^\DFT(\mu_\kappa^\rHF)$.
    To that end, for every $\gamma \in \DO_\kappa{(\mathcal{H})}$ and $\mu \in L^\infty(\R^3)$ we have
    \[ \Omega^\rHF_\kappa(\gamma) \geq \Omega^\DFT_\kappa(\rho_\gamma, \mu) \geq \Omega^\DFT_\kappa(\mu) \]
    and therefore $\Omega^\rHF_\kappa \geq \Omega_\kappa^\DFT(\mu)$ for all $\mu \in L^\infty(\R^3)$.
    In particular, $\lim_{n \to \infty} \Omega_\kappa^\DFT(\mu_n) = \Omega^\rHF_\kappa \geq \Omega_\kappa^\DFT(\mu^\rHF_\kappa)$.
    It then suffices to show that
    \[ \limsup_{n \to \infty} \Omega_\kappa^\DFT(\mu_n) \leq \Omega_\kappa^\DFT(\mu_\kappa^\rHF) \, , \]
    but by Lemma \ref{lsc}, $\Omega_\kappa^\DFT(\mu)$ is upper semicontinuous.

    We have then shown that $\Omega_\kappa^\rHF = \Omega_\kappa^\DFT(\mu_\kappa^\rHF) \geq \Omega_\kappa^\DFT(\mu)$ for all other $\mu \in L^\infty(\R^3)$ as required.
  \end{proof}

  We conclude by proving a parallel lemma to Lemma \ref{ftfmunice} that shows that the maximizers of the dual Levy--Lieb functional above are very nice.

  \begin{lemma}
    \label{fksmunice}

    Let $\mu \in L^\infty(\R^3)$, $v \in L^{5/2}_\mathrm{loc.}(\R^3)$ be confining, $w \in \mathcal{S}(\R^3)$ be even and of positive type, and $s : [0, 1] \to \R$ be an entropy function.
    Suppose that $\Omega^\DFT_\kappa(\mu)$, as defined in \eqref{fdftmudef}, is finite.

    Then $\mu \in H^\sigma(\R^3)$ for all $\sigma \geq 0$ with
    \[ \norm{\mu}_{H^\sigma(\R^3)}^2 \leq - 2 \kappa^3 \Omega^\DFT_\kappa(\mu) \norm{ { \langle - \rangle }^{2\sigma} \, \hat w }_{L^\infty(\R^3)} \, . \]
  \end{lemma}

  \begin{proof}
    Note that by scaling,
    \[ \inf_{\rho \in L^1(\R^3)} \left[ \frac{\kappa^3}{2} \inner{\rho, \rho}_w - \inner{\mu, \rho} \right] = \kappa^{-3} \inf_{\rho \in L^1(\R^3)} \left[ \frac{1}{2} \inner{\rho, \rho}_w - \inner{\mu, \rho} \right] \, . \]
    By Lemma \ref{infizernice}, we then have
    \[ \Omega_\kappa^\DFT(\mu) = \begin{cases} \Tr s_\ast(H_\kappa + \mu(\hat x)) - \frac{1}{2} \int \kappa^{-3} \abs{\hat \mu(k)}^2 \, \hat w(k)^{-1} \, \dd k & \mu \in L^2(\R^3) \\ -\infty & \mathrm{otherwise}  \end{cases} \, . \]
    Suppose that the left-hand side is finite.
    It follows that $\mu \in L^2(\R^3)$ and since $s_\ast \leq 0$,
    \[ \frac{1}{2} \int \abs{\hat \mu(k)}^2 \, \hat w(k)^{-1} \, \dd k \leq - \kappa^3 \Omega_\kappa^\DFT(\mu) \, . \]
    The result then follows from $\norm{\mu}^2_{H^\sigma(\R^3)} = \int \langle k \rangle^{2\sigma} \abs{\hat \mu(k)}^{2} \, \dd k$ and H\"older's inequality.
  \end{proof}

  \section{Semiclassical Results}

  In this section, we state and prove our main technical proposition used to establish Theorem \ref{mainthm}.
  In particular, we establish that the well-known semiclassical convergence of Riesz means of order one can be extended to families of potentials under appropriate compactness assumptions.
  The main result is as follows:

  \begin{proposition}
    \label{semiclassics}

    Let $\{ v_\kappa : \kappa > 0 \} \subseteq L^{5/2}_{\mathrm{loc.}}(\R^3)$
    and $s : [0, 1] \to \R$ be an entropy function such that $s_\ast \in C^2(\R)$.
    Let $H_\kappa = -\kappa^2 \Delta + v_\kappa(\hat x)$.

    Suppose that $\{ \left. v_\kappa \right|_K \}$ is precompact in $L^{5/2}(K)$
    for all compact $K \subseteq \R^3$ and
    there exists $v_{\min} \in L^{5/2}_{\mathrm{loc.}}(\R^3)$ such that $v_{\min}$ is confining, $v_\kappa \geq v_{\min}$ for all $\kappa > 0$, and the compatibility condition
    \[ \iint s_\ast(\abs{p}^2 + v_{\min}(x)) \, \dd x \, \dd p > -\infty \]
    holds.

    Then $\Tr s_\ast(H_\kappa)$ and $\iint s_\ast(\abs{p}^2 + v_\kappa(x)) \, \dd p \, \dd x$ are finite for each $\kappa > 0$, and
    \[ \label{semiclassicsres} \left( \kappa^3 \Tr s_\ast(H_\kappa) - \frac{1}{(2\pi)^3} \iint s_\ast(\abs{p}^2 + v_\kappa(x)) \, \dd p \, \dd x \right) \xrightarrow{\kappa \to 0} 0 \, . \]
  \end{proposition}

  As a precursor, we first show a semiclassical bound for the non-interacting free energy
  in the spirit of the Lieb--Thirring inequality.
  In particular, we introduce the decomposition
  \[ \Tr s_\ast(H_\kappa) = \int_{-\infty}^\infty s_\ast''(t) \, \Tr{\left(t - H_\kappa\right)_+} \, \dd t \]
  which we use heavily throughout the section.

  \begin{lemma}
    \label{semiclassicalbound}
    Let $v \in L^{5/2}_\mathrm{loc.}(\R^3)$ be confining, and $s : [0, 1] \to \R$ be an entropy function
    such that $s_\ast \in C^2(\R)$. Let $H_\kappa = - \kappa^2 \Delta + v(\hat x)$.

    Then $H_\kappa$ is semibounded with a compact resolvent for each $\kappa > 0$, and
    \[ \Tr s_\ast(H_\kappa) \gtrsim \frac{1}{(2\pi \kappa)^3} \iint s_\ast(\abs{p}^2 + v(x)) \, \dd p \, \dd x \, . \]
  \end{lemma}

  \begin{proof}
    The semiboundedness of $H_\kappa$ and the compactness of the resolvent are classical.
    For completeness, we briefly prove semiboundedness.
    A proof of the compactness of the resolvent can be found in Ref. \onlinecite[Theorem XIII.69]{ReedSimon}.

    We define $H_\kappa$ as the unique self-adjoint operator associated to the form
    \[ q_\kappa(\psi) = \int \kappa^2 \abs{\nabla \psi(x)}^2 + v(x) \abs{\psi(x)}^2 \, \dd x \]
    with (form) domain $Q(q_\kappa) = \{ \psi \in H^1(\R^3) : \int \abs{\psi(x)}^2 \, \abs{v(x)} \, \dd x < \infty \}$.

    We then have
    \begin{align}
      \kappa^{-2} q_\kappa(\psi) &\geq \norm{\nabla \psi}^2_{L^2(\R^3)} - \kappa^{-2} \int v_-(x) \abs{\psi(x)}^2 \, \dd x \\[3pt]
      &\geq \norm{\nabla \psi}^2_{L^2(\R^3)} - \kappa^{-2} \norm{v_-}_{L^{5/2}(\R^3)} \norm{\psi}^2_{L^{10/3}(\R^3)} \\
      &\geq \norm{\nabla \psi}^2_{L^2(\R^3)} - c_\textrm{Sobolev}^{6/5} \kappa^{-2} \norm{v_-}_{L^{5/2}(\R^3)} \norm{\psi}^{4/5}_{L^2(\R^3)} \norm{\nabla \psi}^{{6/5}}_{L^2(\R^3)} \\
      &= \norm{\nabla \psi}^2_{L^2(\R^3)} - \left[ c_\textrm{Sobolev}^{3/2} \kappa^{-{5/2}} \norm{v_-}_{L^{5/2}(\R^3)}^{5/4} \norm{\psi}_{L^2(\R^3)} \right]^{4/5} \norm{\nabla \psi}^{{6/5}}_{L^2(\R^3)} \\
      &\gtrsim -\kappa^{-5} \norm{v_-}_{L^{5/2}(\R^3)}^{5/2} \norm{\psi}_{L^2(\R^3)}^2
    \end{align}
    where the last inequality follows from Young's inequality $a^\theta b^{2 - \theta} - a^2 \lesssim_\theta b^2$ and the fact that $v_- \in L^{5 / 2}(\R^3)$
    since $v$ is confining and therefore $v_-$ has compact support.

    By Lemma \ref{entropydecay}, we have $\lim_{h \to \infty} s_\ast(h) = \lim_{h \to \infty} h s_\ast'(h) = 0$ and thus
    \[ s_\ast(h) = \int_h^\infty s_\ast''(t) \, (t - h) \, \dd t = \int_{-\infty}^\infty s_\ast''(t) \, (t - h)_+ \, \dd t \, . \]

    By Tonelli's theorem (which applies due to the concavity of $s_\ast$), it follows that
    \[ \Tr s_\ast(H_\kappa) = \int_{-\infty}^\infty s_\ast''(t) \, \Tr{\left([t - H_\kappa]_+\right)} \, \dd t \, . \]

    By the Lieb-Thirring inequality, we have
    \[ \Tr{\left([t - H_\kappa]_+\right)} \lesssim \frac{1}{(2\pi \kappa)^3} \iint (t - \abs{p}^2 - v(x))_+ \, \dd x \, \dd p \, . \]

    Plugging in this bound and applying Tonelli's theorem once again yields the result.
  \end{proof}

  \begin{lemma}
    \label{entropydecay} Let $s : [0, 1] \to \R$ be an entropy function such that $s_\ast \in C^1(\R)$.
    Then
    \[ \lim_{h \to \infty} s_\ast(h) = \lim_{h \to \infty} h s_\ast'(h) = 0 \, . \]
  \end{lemma}

  \begin{proof}
    Since $s$ is continuous on $[0,1]$, let $m$ be the maximum value of $s$ over $[0, 1]$.
    For each $h > 0$, choose $r_h \in [0, 1]$ such that $s_\ast(h) = h r_h - s(r_h)$.

    Since $r = 0$ is admissible and $s(0) = 0$, we have $s_\ast(h) \leq 0$.
    It follows that $h r_h - s(r_h) \leq 0$,
    and therefore
    \[ 0 \leq r_h \leq \frac{s(r_h)}{h} \leq \frac{m}{h} \xrightarrow{h \to \infty} 0 \, . \]
    We therefore conclude that $r_h \to 0$ as $h \to \infty$.
    Additionally,
    \[ -s(r_h) \leq h r_h - s(r_h) = s_\ast(h) \leq 0 \]
    and therefore, since $s$ is continuous and $s(0) = 0$, we have $s_\ast(h) \xrightarrow{h \to \infty} 0$.

    It remains to prove that \(h s_\ast'(h) \xrightarrow{h \to \infty} 0\).
    To that end, since $s_\ast$ is differentiable, one can show that $s_\ast'(h) = r_h$.
    We therefore have $h s_\ast'(h) = h r_h = s_\ast(h) + s(r_h)$ and thus $h s_\ast'(h) \xrightarrow{h \to \infty} 0$.
  \end{proof}

  As a precursor to Proposition \ref{semiclassics}, we recall the following lemma from Ref. \onlinecite{Nam}.

  \begin{lemma}
    \label{weylfixedmu}

    Let $\epsilon > 0$ and let $v \in L^{3/2+\epsilon}_\mathrm{loc.}(\R^3)$ such that $v_- \in L^{5/2}(\R^3)$.
    Then, as $\lambda \to \infty$,
    \[ \label{weylasmpt} \Tr{\left( -\Delta + \lambda v(\hat x) \right)_-} = \iint \left( \abs{2 \pi k}^2 + \lambda v(x) \right)_- \, \dd k \, \dd x + o\left(\lambda^{5/2}\right) \, , \]
    or equivalently,
    \[ \label{weylsemic} \kappa^3 \Tr{\left( -\kappa^2 \Delta + v(\hat x) \right)_-} \xrightarrow{\kappa \to 0} \frac{1}{(2\pi)^3} \iint \left( \abs{p}^2 + v(x) \right)_- \, \dd p \, \dd x \, . \]
  \end{lemma}

  \begin{proof}
    A proof of \eqref{weylasmpt} is given in \onlinecite[Section 4.2]{Nam}; therefore, it suffices to show that \eqref{weylasmpt} implies \eqref{weylsemic}.
    Indeed, substituting $\lambda^{-1/2} = \kappa$, \eqref{weylasmpt} implies
    \[ \abs*{\kappa^5 \Tr{\left( -\Delta + \kappa^{-2} v(\hat x) \right)_-} - \kappa^5 \iint \left( \abs{2 \pi k}^2 + \kappa^{-2} v(x) \right)_- \, \dd k \, \dd x } \xrightarrow{\kappa \to 0} 0 \, , \]
    and thus
    \[ \abs*{\kappa^3 \Tr{\left( -\kappa^2 \Delta + v(\hat x) \right)_-} - \kappa^3 \iint \left( \abs{2 \pi \kappa \times k}^2 + v(x) \right)_- \, \dd k \, \dd x } \xrightarrow{\kappa \to 0} 0 \, . \]
    The change of variables $2 \pi \kappa \times k = p$ in the integral completes the proof.
  \end{proof}

  We conclude the section with the proof of Proposition \ref{semiclassics}.

  \begin{proof}[Proof of Proposition \ref{semiclassics}]
    We first show finiteness of $\Tr s_\ast(H_\kappa)$ and $\iint s_\ast(\abs{p}^2 + v_\kappa(x)) \, \dd p \, \dd x$.
    Indeed, by Lemma \ref{semiclassicalbound}, and since $s_\ast$ is non-decreasing,
    \[ 0 \geq \kappa^3 \Tr s_\ast(H_\kappa) \geq c_\mathrm{LT} \iint s_\ast(\abs{p}^2 + v_\kappa(x)) \, \dd p \, \dd x \geq c_\mathrm{LT} \iint s_\ast(\abs{p}^2 + v_{\min}(x)) \, \dd p \, \dd x > -\infty \, , \]
    where $c_\mathrm{LT}$ is the best constant in the Lieb--Thirring inequality used in the proof of Lemma \ref{semiclassicalbound}.

    Following the proof of Lemma \ref{semiclassicalbound}, the $\kappa$-dependent expression in \eqref{semiclassicsres} is equal to
    \[ \int s_\ast''(t) \left[ \kappa^3 \Tr \left(H_\kappa - t\right)_- - \frac{1}{(2 \pi)^3} \iint \left( \abs{p}^2 + v_\kappa(x) - t \right)_- \dd p \, \dd x \right] \, \dd t  \equiv \int s_\ast''(t) \, d_\kappa(t) \, \dd t \, . \]

    For each fixed $\kappa > 0$ and $t \in \R$, we have by the Lieb--Thirring inequality
    \[ 0 \leq \kappa^3 \Tr{\left(H_\kappa - t\right)_-} \leq \frac{c_\mathrm{LT}}{(2\pi)^3} \iint (\abs{p}^2 + v_\kappa(x) - t)_- \, \dd p \, \dd x \]
    and therefore it follows by the triangle inequality and the concavity of $s_\ast$ that
    \[ \abs{s_\ast''(t) \, d_\kappa(t)} \leq -s_\ast''(t) \left( \frac{c_\mathrm{LT} + 1}{(2\pi)^3} \right) \iint (\abs{p}^2 + v_\kappa(x) - t)_- \, \dd p \, \dd x \lesssim -s_\ast''(t) \iint (\abs{p}^2 + v_{\min}(x) - t)_- \, \dd p \, \dd x \, . \]

    By Tonelli's theorem, cf. Lemma \ref{semiclassicalbound}, this dominating function integrates to a multiple of
    \[ - \iint s_\ast(\abs{p}^2 + v_{\min}(x)) \, \dd p \, \dd x < \infty \, . \]
    To conclude the proof, it therefore suffices by the dominated convergence theorem to show $d_\kappa(t) \xrightarrow{\kappa \to 0} 0$ for each fixed $t$.
    To that end, let $w_\kappa(t, x) = [v_\kappa(x) - t]_-$.
    By the minmax theorem, it follows that
    \[ d_\kappa(t) \leq \kappa^3 \Tr{(-\kappa^2 \Delta_x - w_\kappa(t, \hat x))_-} - \frac{1}{(2\pi)^3} \iint (\abs{p}^2 + v_\kappa(x) - t)_- \, \dd p \, \dd x \, . \]
    In addition, $w_\kappa(t) \in L^{3/2}(\R^3) \cap L^{5/2}(\R^3)$ for each $\kappa > 0$ since
    \[ \operatorname{supp} w_\kappa(t) \subseteq \overline{\{ v_{\min}(x) \leq t \}} \, , \]
    which is compact as $v_{\min}$ is confining, while $v_\kappa \in L^{5/2}_\mathrm{loc.}(\R^3)$.

    Following the proof of \eqref{weylasmpt} given in \onlinecite[Section 4.2]{Nam}, for every $\epsilon \in (0, 1)$ and non-negative
    radial function $G \in C_{00}^\infty(\R^3)$ with $\norm{G}_{L^2(\R^3)} = 1$, $d_\kappa(t)$ is bounded above by a constant multiple of
    \[ \left[ \left( 1 - \epsilon \right)^{-\frac{3}{2}} - 1 \right] \norm{w_\kappa(t)}_{L^{5/2}(\R^3)}^{5/2} + \epsilon^{-\frac{3}{2}} \norm{w_\kappa(t) - \tilde w_\kappa(t)}_{L^{5/2}(\R^3)}^{5/2} + \kappa^2 \norm{\nabla G}_{L^2(\R^3)}^2 \norm{w_\kappa(t)}_{L^{3/2}(\R^3)}^{3/2} \]
    where $\tilde w_\kappa(t) = G^2 \ast w_\kappa(t)$ and the constant is universal.

    Let $K = \overline{\{ v_{\min}(x) \leq t \}}$.
    By assumption, the set $\{ \left. v_\kappa \right|_{K} \}$ is relatively compact in $L^{5/2}(K)$ and thus the set $\{ \left. [v_\kappa - t] \right|_{K} \}$ is relatively compact in $L^{5/2}(K)$.
    Since the support of $w_\kappa(t)$ is contained in $K$ and $w_\kappa(t) = t - v_\kappa$ on its support, we conclude that $\{ w_\kappa(t) \}$ is relatively compact in $L^{5/2}(\R^3)$.
    In particular, $\norm{ w_\kappa(t) }_{L^{5/2}(\R^3)}$ and $\norm{ w_\kappa(t) }_{L^{3/2}(\R^3)}$ are uniformly bounded.

    Choose $G_n \in L^2(\R^3)$ so that $\left(G_n^2\right)$ is an approximate identity.
    It then follows that $\left( \mathds{1} - G_n^2 \ast - \right) \to 0$ strongly (but not in norm) as $n \to \infty$ in $L^{5 / 2}(\R^3)$.
    Nonetheless, we have
    \[ d_\kappa(t) \lesssim \left[ \left( 1 - \epsilon \right)^{-\frac{3}{2}} - 1 \right] + \epsilon^{-\frac{3}{2}} \sup_{\kappa > 0} \norm{w_\kappa(t) - \tilde w_{\kappa}^{(n)}(t)}_{L^{5/2}(\R^3)}^{5/2} + \kappa^2 \norm{\nabla G_n}_{L^2(\R^3)}^2 \]
    for some constant independent of $\kappa$, $n$, and $\epsilon$.

    To conclude the upper bound, we first take the $\limsup$ as $\kappa \to 0$ of both sides, then take $n \to \infty$, and finally take $\epsilon \to 0$ to find $\limsup_{\kappa \to 0} d_\kappa(t) \leq 0$.
    In particular, although $\left( \mathds{1} - G_n^2 \ast - \right) \not\to 0$ in norm,
    \[ \lim_{n \to \infty} \left( \sup_{\kappa > 0} \norm{w_\kappa(t) - \tilde w_{\kappa}^{(n)}(t)}_{L^{5/2}(\R^3)}^{5/2} \right) = 0 \, , \]
    where $\tilde w_\kappa^{(n)} = G^2_n \ast w_\kappa$.
    This follows from the relative compactness of $\{ w_\kappa(t) \}$ and Lemma \ref{strongcompactuniform}.

    Similarly, again following the proof of \eqref{weylasmpt} given in \onlinecite[Section 4.2]{Nam},
    \[ \liminf_{\kappa \to 0} d_\kappa(t) \gtrsim -\limsup_{\kappa \to 0} \int w_\kappa(t, x)^{3/2} \times \left( \tilde v_\kappa^{(n)}(x) - v_\kappa(x) \right) \, \dd x \]
    where $\tilde v_\kappa^{(n)} = G_n^2 \ast v_\kappa$.

    For each fixed $\kappa$, we decompose this expression as
    \[ \int w_\kappa(t, x)^{3 / 2} \left( w_\kappa(t, x) - \tilde w_\kappa^{(n)}(t, x) \right) \, \dd x + \int w_\kappa(t, x)^{3 / 2} \, \tilde z_\kappa^{(n)}(t, x) \, \dd x \, , \]
    where $z_\kappa(t, x) = [v_\kappa - t]_+$ and $\tilde z_\kappa^{(n)}(t) = G_n^2 \ast z_\kappa(t)$.

    By H\"older's inequality,
    \begin{align}
      \int w_\kappa(t, x)^{3 / 2} \left( w_\kappa(t, x) - \tilde w_\kappa^{(n)}(t, x) \right) \, \dd x &\geq -\norm{w_\kappa(t)}_{L^{5/2}(\R^3)}^{3 / 2} \norm{w_\kappa(t) - \tilde w^{(n)}_\kappa(t)}_{L^{5/2}(\R^3)} \\
      &\geq -\norm{w_\kappa(t)}_{L^{5/2}(\R^3)}^{3 / 2} \sup_{\kappa > 0} \left( \norm{w_\kappa(t) - \tilde w^{(n)}_\kappa(t)}_{L^{5/2}(\R^3)} \right)
    \end{align}
    and, taking $n \to \infty$, this term vanishes and we conclude that
    \[ \liminf_{\kappa \to 0} d_\kappa(t) \gtrsim -\limsup_{\kappa \to 0} \int w_\kappa(t, x)^{3 / 2} \, \tilde z_\kappa^{(n)}(t, x) \, \dd x \, . \]

    The latter term is equal to
    \[ \int \left( G_n^2 \ast \left[ w_\kappa(t)^{3 / 2} \right] \right) (x)\times z_\kappa(t, x) \, \dd x \, . \]
    In particular, we can further ask that the supports of the functions $G_n$ lie in some fixed compact set $\tilde K$.
    It follows that $G_n^2 \ast \left[ w_\kappa(t)^{3 / 2} \right]$ has support in $K' = \overline{K + \tilde K}$, which is compact.

    Since $w_\kappa$ and $z_\kappa$ have disjoint support, this term is also equal to
    \[ \int \left[ \left( G_n^2 \ast \left[ w_\kappa(t)^{3 / 2} \right] \right) (x) - w_\kappa(t, x)^{3 / 2} \right] \times z_\kappa(t, x) \, \dd x \, , \]
    and therefore by H\"older's inequality, it is bounded in absolute value by
    \[ \norm{z_\kappa(t)}_{L^{5 / 2}(K')} \sup_{\kappa > 0} \left( \norm*{\left( \mathds{1} - G_n^2 \ast - \right) \left[ w_\kappa(t)^{3 / 2} \right]}_{L^{5 / 3}(\R^3)} \right) \, . \]

    To conclude, $\{ w_\kappa(t)^{3/2} \}$ is relatively compact in $L^{5/3}(\R^3)$ since $\{ w_\kappa(t) \}$ is relatively compact in $L^{5/2}(\R^3)$ and the power map is continuous.
    $\norm{z_\kappa(t)}_{L^{5/2}(K')}$ is uniformly bounded by the precompactness of $\{ \left. v_\kappa \right|_{K'} \}$.
    By taking $n \to \infty$ we conclude that $\liminf_{\kappa \to 0} d_\kappa(t) \geq 0$.
    We therefore have $\lim_{\kappa \to 0} d_\kappa(t) = 0$, concluding the proof.
  \end{proof}

  \noindent \emph{Remark.} The proof given above is a relatively straightforward extension of the proof given by Nam in \onlinecite[Section 4.2]{Nam}.
  Similar results for fixed potentials $v$ with low regularity have been shown by Frank\cite{Frank} using methods similar to those in \onlinecite{Nam} and more recently by Mikkelsen\cite{Mikkelsen}.

  \begin{acknowledgments}
    The author would like to thank Professor I. M. Sigal and Professor Fabio Pusateri for their guidance and supervision,
    and Professor Almut Burchard for her feedback on a draft of the paper.

    The author acknowledges the support of the Natural Sciences and Engineering Research Council of Canada (NSERC).
  \end{acknowledgments}

  \section*{Tool and Computational Resource Disclosure}
  The author used ChatGPT, provided by OpenAI, to brainstorm possible approaches to mathematical arguments, identify potentially relevant literature, and assist with limited language editing.
  The tools were used to generate suggestions for subsequent evaluation, not as sources of mathematical or bibliographic authority.

  The author independently checked every mathematical argument and verified every citation against the cited source without relying on artificial intelligence.
  The author takes full responsibility for the accuracy, originality, and integrity of the manuscript.

  \section*{Author Declarations}
  \subsection*{Conflict of Interest}
  The author has no conflicts to disclose.
  \subsection*{Author Contributions}
  \textbf{Dominic Shillingford}: Methodology (equal); Writing -- original draft (equal); Writing -- review \& editing (equal).
  \section*{Data Availability}
  Data sharing is not applicable to this article as no new data were created or analyzed in this study.

  \appendix

  \section{Lemmas from Functional Analysis}

  The following lemmas are standard in functional analysis, but included for completeness.

  \begin{lemma}
    \label{operatorfenchel} Let $s : [0, 1] \to \R$ be an entropy function and let $\gamma \in \operatorname{DO}{(\mathcal{H})}$.
    Then for each $\lambda \in \mathbb{R}$ we have
	  \[ \lambda \gamma - s(\gamma) - s_\ast(\lambda) \mathds{1} \geq 0 \, . \]
  \end{lemma}

  \begin{proof}
    By the spectral theorem for compact, self-adjoint operators and the functional calculus of these operators, we have
	  \[ \gamma = \sum_{m = 1}^\infty \mu_m P_m \ \ \text{and}\ \ s(\gamma) = \sum_{m = 1}^\infty s(\mu_m) P_m \]
	  for eigenvalues $\mu_m \in [0, 1]$ and projectors $P_m$ such that $\sum_{m = 1}^\infty P_m u = u$ for each fixed $u \in \mathcal{H}$.

    It follows that for each $u \in \mathcal{H}$ we have
	  \[ \inner{ u, [\lambda \gamma - s(\gamma) - s_\ast(\lambda)] u } = \sum_{m = 1}^\infty (\lambda \mu_m - s(\mu_m) - s_\ast(\lambda) ) \, \inner{ u, P_m u } \, . \]
	  Since $P_m$ is a projection, $\inner{ u, P_m u } = \norm{ P_m u }^2 \geq 0$ while $\lambda \mu_m - s(\mu_m) - s_\ast(\lambda) \geq 0$
	  by definition (Fenchel's inequality).
	  The sum is therefore a sum of non-negative terms, completing the proof.
  \end{proof}

  \begin{lemma}
    \label{strongcompactuniform} Let $X$ and $Y$ be Banach spaces and let $(A_n) \subseteq \mathcal{B}(X, Y)$ be a sequence such that $A_n x \to 0$ for every $x \in X$.
    Then $A_n x \to 0$ uniformly for $x$ in relatively compact subsets of $X$.
    In other words, if $K \subseteq X$ is relatively compact, then
    \[ \sup_{x \in K} \norm{A_n x}_Y \xrightarrow{n \to \infty} 0 \, . \]
  \end{lemma}

  \begin{proof}
    By the uniform boundedness principle, there is a constant $M < \infty$ such that $\norm{A_n}_{\mathcal{B}(X, Y)} \leq M$ for all $n \in \mathbb{N}$.
    Suppose the claim fails for some relatively compact $K \subseteq X$.
    Then there exist $\epsilon > 0$ and, perhaps passing to a subsequence,
    $x_n \in K$ such that $\norm{A_n x_n}_Y \geq \epsilon$ for all $n \in \mathbb{N}$.
    Since $K$ is relatively compact, perhaps passing to a further subsequence, there exists $x \in X$ such that $x_n \to x$ in $X$.
    It follows that
    \begin{align*}
      \norm{A_n x_n}_Y &\leq \norm{A_n(x_n - x)}_Y + \norm{A_n x}_Y \\
      &\leq M \norm{x_n - x}_X + \norm{A_n x}_Y \xrightarrow{n \to \infty} 0,
    \end{align*}
    which contradicts the lower bound above.
  \end{proof}

  \begin{lemma}
    \label{lsc} Let $f : \mathbb{R} \to [0,\infty]$ be lower semicontinuous.
    Let $H$ be self-adjoint, semibounded, and have compact resolvent.
    Then the map
    \[ A \mapsto \Tr f(H + A) \]
    is lower semicontinuous on the bounded self-adjoint operators equipped with the operator norm topology.
  \end{lemma}

  \begin{proof}
    Let $(\lambda_n(A))$ denote the eigenvalues of $H+A$, counted with multiplicity.
    The min--max principle gives
    \[ \abs{\lambda_n(A)-\lambda_n(B)} \leq \norm{A-B}_{\mathrm{op}} \]
    for all bounded self-adjoint $A$ and $B$.
    Thus $A \mapsto f(\lambda_n(A))$ is lower semicontinuous for every $n$.
    Since $f \geq 0$, the spectral theorem yields
    \[ \Tr f(H+A) = \sup_{N \geq 1} \sum_{n=1}^N f(\lambda_n(A)) \, , \]
    which is lower semicontinuous as a supremum of lower semicontinuous functions.
  \end{proof}

  \bibliography{paper}

\end{document}